\documentclass{article}

\usepackage{amssymb,amsmath,amsthm}
\usepackage{cancel}
\usepackage{graphics}
\usepackage{diagbox}
\usepackage{graphicx}
\usepackage{caption}
\usepackage{hyperref}
\usepackage{enumerate}
\usepackage{color}         
\usepackage{overpic}
\usepackage{subfig}
\usepackage{algorithm}
\usepackage{algpseudocode}
\usepackage[margin=0.75in]{geometry}
\usepackage[title]{appendix}
\ifpdf
  \DeclareGraphicsExtensions{.eps,.pdf,.png,.jpg}
\else
  \DeclareGraphicsExtensions{.eps}
\fi

\theoremstyle{plain} 
\newtheorem{theorem}{Theorem}[section]
\newtheorem*{theorem*}{Theorem}
\newtheorem{lemma}[theorem]{Lemma}

\newtheorem{corollary}[theorem]{Corollary}

\theoremstyle{definition}
\newtheorem{remark}[theorem]{Remark}
\newtheorem{definition}[theorem]{Definition}

\newtheorem*{Assumptions}{Assumptions}

\DeclareMathOperator*{\argmin}{arg~min}

\DeclareMathOperator{\supp}{supp}

\newcommand{\bfx}{{\bf x}}
\newcommand{\bfb}{{\bf b}}
\newcommand{\bfu}{{\bf u}}
\newcommand{\bfv}{{\bf v}}

\newcommand{\bfy}{{\bf y}}
\newcommand{\bfe}{{\bf e}}
\newcommand{\bfz}{{\bf z}}
\newcommand{\bfw}{{\bf w}}
\newcommand{\bfr}{{\bf r}}

\newcommand{\bfone}{{\bf 1}}

\newcommand{\ER}{Erd\H{o}s - R\`enyi }

\def\3bar{{|\hspace{-.02in}|\hspace{-.02in}|}}

\title{Compressive Sensing for cut improvement and local clustering.}
\author{Ming-Jun Lai\footnote{mjlai@uga.edu. Department of Mathematics,
University of Georgia, Athens, GA 30602.
This research is partially supported by 
the National Science Foundation under the grant \#DMS 1521537. }
\and
Daniel Mckenzie \footnote{mckenzie@math.ucla.edu. Department of 
Mathematics, University of California, Los Angeles, CA 155505. The financial assistance of the National Research Foundation of South Africa 
(NRF) towards this research is hereby acknowledged. Opinions expressed and conclusions arrived at, are those of the author and not 
necessarily to be attributed to the NRF. This research was conducted while this author was a graduate student at the University of Georgia and he gratefully acknowledges support and encouragement received from the Math Department of UGA.}}
\date{\today}

\begin{document}

\maketitle

\begin{abstract}
We show how one can phrase the cut improvement problem for graphs as a sparse recovery problem, whence one can use algorithms originally developed for use in compressive sensing (such as {\tt SubspacePursuit} or {\tt CoSaMP}) to solve it. We show that this approach to cut improvement is fast, both in theory and practice and moreover enjoys statistical guarantees of success when applied to graphs drawn from probabilistic models such as the Stochastic Block Model. Using this new cut improvement approach, which we call {\tt ClusterPursuit}, as an algorithmic primitive we then propose new methods for local clustering and semi-supervised clustering, which enjoy similar guarantees of success and speed.  Finally, we verify the promise of our approach with extensive numerical benchmarking.  
\end{abstract}

\begin{center}
\end{center}


\section{Introduction}
Finding clusters is a problem of primary interest when analyzing graphs. 
This is because vertices which are in the same cluster 
can reasonably be assumed to have some latent similarity. Thus, clustering can be used to find communities in social networks \cite{GN02, TMP12, WWMBN14} or deduce political affiliation from a network of blogs \cite{AG05}. Moreover, even data sets which are not presented as graphs can profitably be studied by first creating an auxiliary graph ({\em eg.} a $K$- or $\epsilon$-nearest-neighbors graph) and then applying graph clustering techniques. This has been successfully applied to image segmentation \cite{SM00,MOV12}, image classification \cite{JME18} and natural language processing \cite{D01}. \\

\begin{figure}[!h]
\centering
\minipage{.32\textwidth}
  \includegraphics[width =0.9\linewidth]{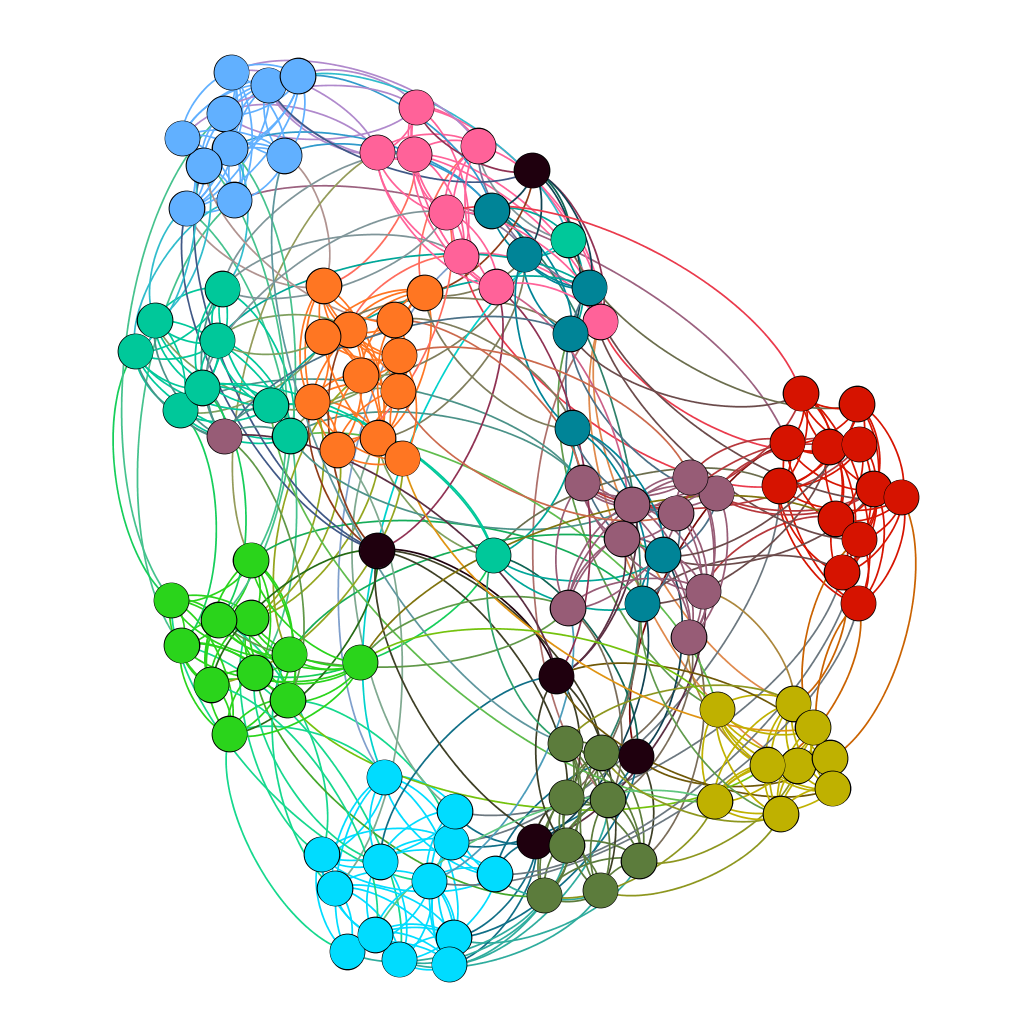}
\endminipage\hfill
\minipage{.32\textwidth}
\includegraphics[width =1.1\linewidth]{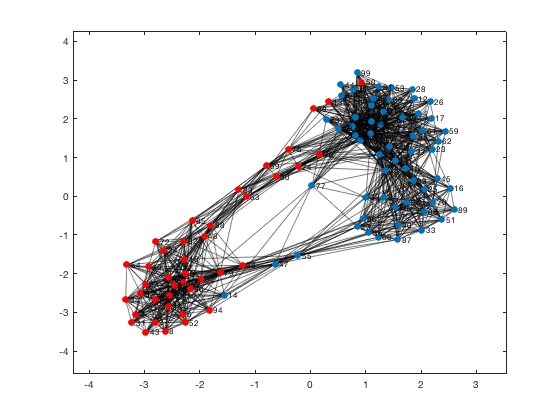}
\endminipage\hfill
\minipage{.32\textwidth}%
  \includegraphics[width = \linewidth]{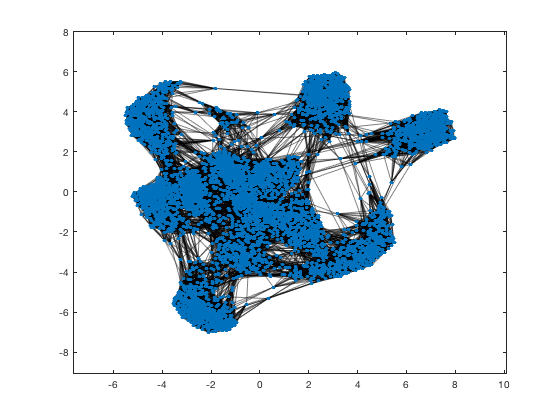}
\endminipage
\caption{{\em Left:} the College Football graph of \protect{\cite{GN02}}. Vertices represents colleges fielding (American) football teams in the 2000 season. Vertices are connected if the respective teams played each other during the regular season. Clusters correspond to the various conferences in which teams play. Note that there are five schools, denoted in black, which are ``independents'' {\em ie} they are not affiliated with any conference. These can be thought of as background vertices. {\em Middle:} Senate co-voting for the 97th Congress, created using data from \protect{\cite{L19}}. Vertices represent Senators and are connected if the respective Senators cast the same vote on a majority of bills. The two large clusters correspond to the two major American political parties. Notice how the blue cluster can be visually subdivided into two sub-clusters. {\em Right:} The OptDigits dataset consists of $5620$ grayscale images of handwritten digits 0--9 of size $8\times 8$. We discuss how to turn this into a graph in \S \protect{\ref{sec:NumericalExperiments}}. Note that as there are ten digits, we expect this graph to have ten disjoint clusters.}
\label{fig:Examples_of_clusters}
\end{figure}

We shall informally think of a cluster as a subset of vertices, $C\subset V$ with many edges between vertices in $C$, and few edges to the rest of the graph, $C^{c}$. See Figure \ref{fig:Examples_of_clusters} for a few examples. While some graphs may allow a neat partitioning into disjoint clusters (for example the OptDigits graph in Figure \ref{fig:Examples_of_clusters}), for many graphs this is not the case. Some graphs may contain {\em background vertices}, that is, vertices which do not belong to any cluster (see the College Football graph in Figure \ref{fig:Examples_of_clusters}). Alternatively, graphs may exhibit clusters at multiple scales (See the Senate Co-voting graph in Figure \ref{fig:Examples_of_clusters}). In many cases, one has certain {\em a priori} information that could be used to improve clustering. For example in the OptDigits graph, we may know that some small subset, $\Gamma\subset V$, all represent images of ones. It is reasonable to assume that algorithms which incorporate this additional information (usually referred to as semi-supervised algorithms) will perform better than ones which do not. With this in mind, it is convenient to appeal to the following taxonomy of clustering algorithms:

\begin{enumerate}
\item {\em Global clustering algorithms} assign every vertex to one of $k$ clusters, where the clusters may or may not be disjoint. Algorithms for this problem may be unsupervised (for example {\tt SpectralClustering} \cite{SM00,NJW02} or {\tt GenLouvain} \cite{MFFP11})  or semisupervised (for example the auction dynamics approach of \cite{JME18}, or the regional force based methods of \cite{YT18}). This is appropriate for graphs such as the OptDigits graph of Figure \ref{fig:Examples_of_clusters}, where one expects a clear partition of the vertices into clusters.
	\item {\em Local clustering algorithms}\footnote{Also known as cluster extraction algorithms in the statistics literature} take as input a small set of ``seed vertices'', $\Gamma\subset V$ and return a good cluster containing $\Gamma$. Algorithms for local clustering are not confounded by background vertices, as they are not required to assign them to a cluster. One can further sub-divide local clustering algorithms into strongly and weakly local clustering algorithms. Strongly local algorithms, for {\tt Nibble} \cite{ST04,ST13}, {\tt PPR-Grow} \cite{A06} or {\tt CapacityReleasingDiffusion} \cite{W17}, are characterized by having run time proportional to the size of the cluster found. This is advantageous when the cluster in question has much fewer vertices than the graph as a whole. Weakly local algorithms are characterized as having run time proportional to the size of $G$. In practice they are frequently faster than strongly local algorithms when finding large or moderately large clusters. We note that both kinds of local clustering algorithms may take as input a scale parameter, which dictates the size of the output cluster returned. This is useful when the graph at hand contains clusters at multiple scales, such as the Senate co-voting graph of  Figure \ref{fig:Examples_of_clusters}.
\item {\em Cut improvement algorithms}(cf. \cite{A06}, \cite{O14}, \cite{V16}) take as input a cut, or subset $\Omega\subset V$, which one can think of as an approximation to a cluster $C$, and refine it to produce a better approximation. Often cut improvement algorithms are run on the output of a local clustering algorithm to improve the quality of the output.\\
\end{enumerate}

The central contribution of this paper is a new cut improvement algorithm which we call {\tt ClusterPursuit}, that 
phrases the cut improvement problem as a sparse recovery problem. 
We pair this with a simple local clustering algorithm which we call Random Walk Thresholding or {\tt RWThresh} to obtain a two-stage weakly local clustering algorithm that we shall refer to as {\tt CP+RWT}. One can iterate this algorithm to find all clusters in a graph; we call this procedure iterated {\tt CP+RWT} or {\tt ICP+RWT}. After presenting some mathematical preliminaries and outlining the assumptions we place on generative models of graphs in \S \ref{sec:Preliminaries}, we derive the {\tt ClusterPursuit} algorithm in \S \ref{sec:ClusterPursuit} and prove that, given a cut $\Omega$ satisfying $|C_{1}\bigtriangleup\Omega|/|C_{1}| = O(1)$ {\tt ClusterPursuit} returns $C_{1}^{\#}$ satisfying  $|C_{1}\bigtriangleup C_{1}^{\#}|/|C_{1}| = o(1)$. Here, $C_{1}$ denotes the smallest cluster in the graph. In \S \ref{sec:RWThresh} we discuss the {\tt RWThresh} algorithm, and show that given a small set of seed vertices, $\Gamma\subset C_1$, it is capable of finding an $\Omega$ satisfying $|C_{1}\bigtriangleup\Omega|/|C_{1}| = O(1)$. This leads naturally to guarantees of success for the two-stage local clustering algorithm {\tt CP+RWT}, which we present in \S \ref{sec:CPRWT}. In \S \ref{sec:ICPRWT} we briefly discuss {\tt ICP+RWT} while in \S \ref{sec:Complexity} we show that {\tt CP+RWT} and {\tt ICP+RWT} enjoy a computational complexity of $O(n d_{\max}\log(n))$ where $d_{\max}$ is the largest vertex degree in the graph. In \S \ref{sec:LiteratureReview} we survey the literature and compare our work with relevant recent work in the area, while in \S \ref{sec:SBM} we show that a popular generative model of graphs with communities, namely the stochastic block model, satisfies the assumptions outlined in \S \ref{sec:Preliminaries}. Finally, we complement theoretical insight with experimental results in \S \ref{sec:NumericalExperiments}. In the interest of reproducibility, we make our code available at: \url{danielmckenzie.github.io}. 
  
 \section{Preliminaries}
 \label{sec:Preliminaries}

\subsection{Graph Notation and Definitions}
We restrict our attention to finite, simple, undirected graphs $G = (V,E)$, possibly with non-negative edge weights. We identify the vertex set $V$ with the integers $[n]:= \{1,\ldots, n\}$ and denote an edge between vertices $i$ and $j$ as $\{i,j\}\in E$. The (possibly weighted) adjacency matrix of $G$ will be denoted as $A$. By $d_i$ we mean the degree of the $i$-th vertex, computed as $d_i = \sum_{j} A_{ij}$. For any $S\subset V$ define $\text{vol}(S) = \sum_{i\in S} d_i$. For quantities such as $d_i$ (and later $\lambda_i$) that are indexed by $i\in [n]$,  let $d_{\max} := \max_i d_i$ and similarly $d_{\min} := \min_{i} d_i$. Denote by $D$ the diagonal matrix whose $(i,i)$ entry is $d_i$. By ``cluster'' we shall mean a subset of vertices, $C\subset V$, that is well-connected but sparsely connected to the rest of the graph. If a graph has clusters we shall refer to them as $C_{1},\ldots, C_{k}$. We define $n_a := |C_a|$ and assume that the clusters are ordered by size, so that $n_{1}\leq n_2\leq \ldots, \leq n_{k}$. We reserve the letters $a$ and $b$ for indexing clusters, while $i$ and $j$ will index vertices. 

\begin{definition}[Laplacians of graphs]
The normalized, random walk Laplacian is defined as $L = I - D^{-1}A$. We shall simply refer to it as 
\emph{the Laplacian}. The normalized, symmetric Laplacian is: $L^{\text{sym}} := I - D^{-1/2}AD^{-1/2}$. 
\label{def:Laplacians}
\end{definition} 

Recall the following elementary result in spectral graph theory (see \cite{L07}, for 
example, for a proof):

\begin{theorem}
\label{thm:EvecFormBasis}
Let $C_{1},\ldots, C_{k}$ denote the connected components of a graph $G$. Then the cluster 
indicator vectors
$\mathbf{1}_{C_1},\ldots, \mathbf{1}_{C_k}$ form a basis for the kernel of $L$.
\end{theorem}

Suppose that $G$ has clusters $C_{1},\ldots, C_{k}$. By definition, clusters have few edges 
between them, and so it is useful  to write $G$ as the union of two edge-disjoint subgraphs, 
defined as follows: let $G^{\text{in}} = (V,E^{\text{in}})$ have only edges between vertices in 
the same cluster, while $G^{\text{out}}= (V,E^{\text{out}})$ consist only of edges between 
vertices in different clusters. We emphasize that this is a theoretical construction, as in 
practice we of course cannot ascertain whether two vertices are in the same cluster without 
first solving the clustering problem, which is precisely what we are trying to do. Denote by 
$A^{\text{in}}$ and $L^{\text{in}}$ (resp. $A^{\text{out}}$ and $L^{\text{out}}$) the adjacency 
matrix and Laplacian of $G^{\text{in}}$ (resp. $G^{\text{out}}$). Similarly,  $d_{i}^{\text{in}}$
 (resp. $d^{\text{out}}_{i}$) shall denote the degree of the vertex $i$ in the graph 
$G^{\text{in}}$ (resp. $G^{\text{out}}$). For future reference we define the random walk transition matrices $P = AD^{-1}$ and $N:= D^{-1/2}AD^{-1/2}$.  We note that the spectra of $P,N, A, L$ are related:


\begin{lemma}
For any matrix $B$ with real eigenvalues let $\lambda_i(B)$ denote the $i$-th smallest eigenvalue, counted with multiplicity. Then $\lambda_i(L) = \lambda_i(L^{\text{sym}})$ while $\lambda_{n-i}(N) = \lambda_{n-i}(P) = 1 - \lambda_i(L)$
\label{lemma:Evals_of_L_Lsym_P}
\end{lemma}
\begin{proof}
Observe that $L = D^{-1/2}L^{\text{sym}}D^{1/2}$, hence $L$ and $L^{\text{sym}}$ have the same spectrum. Similarly $P = D^{1/2}\left(I - L^{\text{sym}}\right)D^{-1/2}$ hence $P$ and $N = I - L^{\text{sym}}$ have the same spectrum. Thus if $\lambda$ is the $i$-th smallest eigenvalue of $L^{\text{sym}}$ it is the $i$-th largest (and hence the $(n-i)$-th smallest) eigenvalue of $I - L^{\text{sym}}$. 
\end{proof}

For any $S\subset V$, we denote by $G_{S}$ the induced sub-graph with vertices $S$ and 
edges all $\{i,j\}\in E$ with $i,j\in S$. By $A_{G_{S}}$ (resp. $L_{G_{S}}$) we mean the adjacency matrix (resp. Laplacian) of the graph $G_S$. Note that $L_{G_S}$ is not a submatrix of $L$!
For any  $S\subset [n]$ we define an {\it indicator vector} $\mathbf{1}_{S}\in \mathbb{R}^{n}$ by 
$(\mathbf{1}_{S})_{i} = 1$ if 
$i\in S$ and $(\mathbf{1}_{S})_{i} = 0$ otherwise. 
$|S|$ will always denote the cardinality of $S$. For any matrix 
$B$, by $B_{S}$ we mean the submatrix of $B$ consisting of the columns 
$b_i$ for all $i\in S$. 

\subsection{Compressive Sensing}
\label{subsection:CS}
Recall for any $\bfx\in\mathbb{R}^{n}$,  $\|\bfx\|_{0}:= |\text{supp}(\bfx)| = |\{ i: \ x_{i} \neq 0\}|$ is the sparsity of $\bfx$. 
If $\|\bfx\|_{0} \ll n$ we say that $\bfx$ is {\em sparse}. Cand\'{e}s, Donoho and 
their collaborators in \cite{D06, CRT06} pioneered the study 
of compressive sensing, which offers theoretical analysis 
and algorithmic tools for finding sparse solutions to linear systems $\Phi\bfx = \bfb$, for example by solving the  minimization problem:
\begin{equation}
\text{argmin} \|\Phi\bfx - \bfy \|_{2} \text{ subject to } \|\bfx \|_{0} \leq s,  
\label{eq:CSProblem}
\end{equation}
where  $\Phi\in\mathbb{R}^{m\times n}$ is referred to as the \emph{sensing matrix}. Typically, it is assumed that $m \leq n$ although this will not be the case in this paper.  There are many algorithms available to 
solve Problem \eqref{eq:CSProblem}, but the one we shall focus on is the {\tt SubspacePursuit} algorithm introduced in 
\cite{DM09}.

\begin{algorithm}
\caption{{\tt SubspacePursuit}, as presented in \cite{DM09}}
\label{algorithm:SP}
Input variables: measurement matrix $\Phi$, measurement vector $\bfy$, sparsity parameter $s$ and number of iterations $J$.
\begin{algorithmic}
\State Initialization:
	\State (1) $S^{(0)} = \mathcal{L}_{s}(\Phi^\top \bfy)$.
	\State (2) $\bfx^{(0)} = \argmin_{\bfz\in\mathbb{R}^{N}}\{ \|\bfy - \Phi\bfz\|_{2}: \ \supp(\bfz) \subset S^{(0)}\}$
	\State (3) $\bfr^{(0)} = \bfy - \Phi\bfx^{(0)}$
	\For{$j = 1:J$}  
	\State (1) $\hat{S}^{(j)} = S^{(j-1)}\cup\mathcal{L}_{s}\left(\Phi^\top \bfr^{(j-1)}\right)$
	\State (2) $\displaystyle \bfu = \argmin_{\bfz\in\mathbb{R}^{N}}\{  \|\bfy - \Phi\bfz\|_{2} :\ \supp(\bfz) \subset \hat{S}^{(j)}\}$
	\State (3) $S^{(j)} = 	\mathcal{L}_{s}(\bfu)$ and $\bfx^{(j)} = \mathcal{H}_{s}(\bfu)$
	\State (4) $\bfr^{(j)} = \bfy - \Phi\bfx^{(j)}$
	\EndFor
\end{algorithmic}
\end{algorithm}

Here $\mathcal{L}_{s}(\cdot)$ and $\mathcal{H}_{s}(\cdot)$ are thresholding operators:
\begin{align*}
& \mathcal{L}_{s}(\bfv) := \{i\in [n]: \ v_i \text{ among } s \text{ largest-in-magnitude entries in } \bfv \}\\
& \mathcal{H}_{s}(\bfv)_{i} := \left\{\begin{array}{cc} v_i & \text{ if } i \in \mathcal{L}_{s}(\bfv) \\ 0 & \text{ otherwise. } \end{array} \right.
\end{align*}
In quantifying whether \eqref{eq:CSProblem} has a unique solution, the following constant is often used (see \cite{FR13}) 
\begin{definition}
The $s$ Restricted Isometry Constant ($s$-RIC) of $\Phi\in\mathbb{R}^{m\times n}$, written $\delta_{s}(\Phi)$, is defined to be the smallest value of $\delta > 
0$ such that, for all $\bfx\in\mathbb{R}^{n}$ with $\|\bfx\|_{0} \leq s$, we have:
\begin{equation*}
(1-\delta)\|\bfx\|_{2}^{2} \leq \|\Phi\bfx\|_{2}^{2} \leq (1+\delta)\|\bfx\|_{2}^{2}. 
\end{equation*}
If $\delta_{s}(\Phi) < 1$ we often say that $\Phi$ has the \emph{Restricted Isometry Property} (RIP).
\end{definition}
One of the reasons for the remarkable usefulness of compressive sensing is its robustness to error, both additive (\emph{i.e.} in $\bfy$) and multiplicative (\emph{i.e.} in $\Phi$).  More precisely, suppose that a signal $\hat{\bfy} = \hat{\Phi}\bfx^{*}$ is acquired, but that we do not know the sensing matrix $\hat{\Phi}$ exactly. Instead, we have access only to $\Phi = \hat{\Phi} + M$, for some small perturbation $M$. Suppose further that there is some noise in the measurement process, so that the signal we actually receive is $\bfy = \hat{\bfy} + \bfe$. Can one hope to approximate a sparse vector $\bfx^{*}$ from $\bfy$, given only $\Phi$? This question is answered in the affirmative way by several authors, starting with the work of \cite{HS10}. For {\tt SubspacePursuit}, we have the following result (cf. \cite{Li16}):

\begin{theorem}
\label{thm:PerturbedSP}
Let $\bfx^{*}$, $\bfy$ $\hat{\bfy}$, $\Phi$ and $\hat{\Phi}$ be as above and suppose that $\|\bfx^{*}\|_{0} \leq s$. For any $t\in [n]$, let $\delta_{t} := \delta_{t}(\Phi)$. Define the following constants:
\begin{equation*}
 \epsilon_{\bfy} := \|\bfe\|_{2}/\|\hat{\bfy}\|_{2} \text{ and } \epsilon^{s}_{\Phi} = \|M\|_{2}^{(s)}/\|\hat{\Phi}\|_{2}^{(s)}
\end{equation*}
where for any matrix $B$, $\|B\|_{2}^{(s)} := \max\{ \|B_{S}\|_{2}: \ S\subset [n] \text{ and } |S| = s\}$. Define further: 

\begin{equation*}
\rho = \frac{\sqrt{2\delta_{3s}^{2}(1+\delta_{3s}^{2})}}{1 - \delta_{3s}^{2}} \quad \text{ and } \quad \tau = \frac{(\sqrt{2} + 2)\delta_{3s}}{\sqrt{1 - \delta_{3s}^{2}}}(1 - \delta_{3s})(1 - \rho) + 
\frac{2\sqrt{2}+1}{(1 - \delta_{3s})(1-\rho)} 
\end{equation*}
Assume $\delta_{3s} \leq 0.4859$ and let $\bfx^{(m)}$ be the output of {\tt 
SubspacePursuit} applied to Problem \eqref{eq:CSProblem} after $m$ iterations. Then: 
\begin{equation*}
\frac{\|\bfx^{*} - \bfx^{(m)}\|_{2}}{\|\bfx^{*}\|_{2}} \leq \rho^{m} + \tau\frac{\sqrt{1 + \delta_{s}}}{1 - \epsilon^{s}_{\Phi}}(\epsilon^{s}_{\Phi} + \epsilon_{\bfy}).
\end{equation*}
\end{theorem}
\begin{proof}
This is Corollary 1 in \cite{Li16}. Note that our convention on hats is different to theirs --- our $\Phi$ is their $\hat{\Phi}$, hence our $\rho$ is their $\hat{\rho}$ and so on.
\end{proof}

Next it is easy to obtain bounds on the quantity $\|B\|_{2}^{(s)}:= \max_{S\subset [n] \atop |S| = s}\|B_{S}\|_{2}$:
\begin{lemma}
For any matrix $B$ and any $2 \leq s \leq n$ we have that $\sigma_{s-1}(B) \leq \|B\|_{2}^{(s)} \leq \sigma_{\max}(B) = \|B\|_2$, where $\sigma_{j}(B)$ denotes the $j$-th smallest singular value of $B$.
\label{lemma:s_norm_bound_with_eigenvalues}
\end{lemma}
\begin{proof}
Observe that, for any matrix $B$,
\begin{equation*}
\|B\|_{2}^{(s)}= \max_{S\subset [n] \atop |S| = s}\|B_{S}\|_{2} = \max_{S\subset [n] \atop |S| = s}\sigma_{\max}(B_{S}), 
\end{equation*}
where $\sigma_{\max}(B_{S})$ denotes the maximum singular value of $B_S$. Because $\sigma_{\max}(B_{S}) = \sigma_{s}(B_S)$, by the interlacing theorem for singular values (cf. \cite{T72})  $\sigma_{s-1}(B) \leq \sigma_{\max}(B_{S}) \leq \sigma_{\max}(B)$.
\end{proof}

\subsection{The Data Model}
\label{sec:DataModel}
For conceptual clarity, we shall take an asymptotic viewpoint, and consider graphs $G\in 
 \mathcal{G}_{n}$ as $n\to\infty$. Note that the graphs under consideration may be weighted or unweighted. We say that a graph property $P$ holds {\em almost surely} for $\mathcal{G}_n$ if the probability of a $G$ drawn from $\mathcal{G}_n$ {\em not having} $P$ is $o(1)$. 

\begin{Assumptions} Suppose that  there exist $\epsilon_{i}  = o(1)$ as $n\to \infty$
for $i=1,2, 3$ such that for all $G\in \mathcal{G}_n$:
\end{Assumptions}

\begin{enumerate}
\item [(A1)] $V = C_1\cup\ldots\cup C_k$ where the $C_{a}$ are disjoint, planted clusters and $k$ is $O(1)$ as $n\to \infty$.
\item[(A2)] For all $a\in [k]$ we have that $\lambda_{2}(L_{G_{C_a}}) \geq 1-\epsilon_1$ 
and $\lambda_{n_a}(L_{G_{C_a}}) \leq 1+\epsilon_1$ almost surely.
\item[(A3)] letting $r_{i}:= d^{\text{out}}_i/d^{\text{in}}_i$,  $r_{i} \leq \epsilon_2$ 
for all $i\in [n]$ almost surely.
\item[(A4)] If $d^{\text{in}}_{\text{av}} := \mathbb{E}[d^{\text{in}}_{i}]$ then  
$d^{\text{in}}_{\max} \leq (1+\epsilon_{3}) d^{\text{in}}_{\text{av}}$ and 
$d^{\text{in}}_{\min} \geq (1-\epsilon_{3}) d^{\text{in}}_{\text{av}}$ almost surely.
\end{enumerate}

Note that we can think of (A1)--(A4) as ``regularity'' requirements for graphs; 
as they insist that degrees do not  
vary too wildly, and that the eigenvalues are well behaved. In \S \ref{sec:SBM} we verify that a common 
model of unweighted graphs with clusters---the stochastic block model---satisfies these 
assumptions, so they are certainly not too restrictive. It seems probable (and indeed supported by the numerical evidence of \S \ref{sec:NumericalExperiments}) that reasonable models of random weighted graphs satisfy these properties too, although we leave this for future work. 

\section{The {\tt ClusterPursuit} Algorithm}
\label{sec:ClusterPursuit}
The motivation for our algorithm is the following observation. Suppose for a moment that one had access to $L^{\text{in}}$. Suppose further that one is given a cut $\Omega$ ``near'' a cluster of interest, $C_a$, which we shall take quantitatively to mean that $|C_a\bigtriangleup \Omega| = \epsilon |C_a|$, where $\bigtriangleup$ denotes the symmetric difference, i.e $C \bigtriangleup  \Omega= (C\backslash \Omega) \cup (\Omega\backslash C)$ 
and $\epsilon \in (0,1)$. Letting $U = C_a\setminus \Omega$ and $W = \Omega\setminus C_a$ one observes that:
\begin{align*}
	& \bfone_{\Omega} = \bfone_{C_a} + \bfone_{W} - \bfone_{U} \\
	\implies & L^{\text{in}}\bfone_{\Omega} =  L^{\text{in}}\bfone_{C_a} + L^{\text{in}}\left(\bfone_{W} - \bfone_{U}\right) \\
	\implies  & L^{\text{in}}\bfone_{\Omega} = 0 + L^{\text{in}}\left(\bfone_{W} - \bfone_{U}\right) \quad \text{ (by Theorem \ref{thm:EvecFormBasis})} \\
	\implies & \bfy^{\text{in}} = L^{\text{in}}\left(\bfone_{W} - \bfone_{U}\right) \quad \text{ (if  $\bfy^{\text{in}} :=  L^{\text{in}}\bfone_{\Omega}$)}
\end{align*}
Solving the linear system $\bfy^{\text{in}} = L^{\text{in}}\bfx$ is unlikely to yield $\bfx = \bfone_{W} - \bfone_{U}$, as $L^{\text{in}}$ has a large kernel (Theorem \ref{thm:EvecFormBasis}). However, Theorem \ref{thm:Unique_Sol_Unperturbed} will show that one may recover $\bfone_{W} - \bfone_{U}$ as the solution to the sparse recovery problem:
\begin{equation}
\argmin_{\bfx\in\mathbb{R}^{n}} \left\{ \|L^{\text{in}}\bfx - \bfy^{\text{in}}\|_{2}: \ \|\bfx\|_{0} \leq s \right\}
\label{eq:UnperturbedSparseRecovery}
\end{equation}
where $s \approx |C_a\bigtriangleup \Omega|$. Of course, one will not in practice have access to $L^{\text{in}}$, only $L$. Thus one needs to consider a perturbed version of \eqref{eq:UnperturbedSparseRecovery}:

\begin{equation}
\argmin_{\bfx\in\mathbb{R}^{n}} \left\{ \|L\bfx - \bfy\|_{2}: \ \|\bfx\|_{0} \leq s \right\}
\label{eq:PerturbedSparseRecovery}
\end{equation}
where $\bfy = L\bfone_{\Omega}$. Theorem \ref{thm:SCP_ClusterPursuitWorks} will show that the solution 
$\bfx^{\#}$ to the minimization problem  \eqref{eq:PerturbedSparseRecovery} found by {\tt SubspacePursuit}  is a good enough approximation to $\bfone_{U} - \bfone_{W}$, hence one may infer $U$ and $W$ from the signed support of $\bfx^{\#}$. Clearly, if one knows $\Omega, U$ and $W$ one may reconstruct $C_a$ as $C_a = \left(\Omega\setminus W\right)\cup U$. This is the essence of {\tt ClusterPursuit}, which we present as Algorithm \ref{alg:ClusterPursuit}.

\begin{algorithm}
   \caption{{\tt ClusterPursuit}}
   \label{alg:ClusterPursuit}
\begin{algorithmic}
   \State {\bfseries Input:} Adjacency matrix $A$, initial cut $\Omega$, estimate $s \approx |\Omega\bigtriangleup C_a|$ and $R\in [0,1)$. 
   \State (1) Compute $L = I - D^{-1}A$ and $\bfy = L\bfone_{\Omega}$. 
   \State (2) Let $\bfx^{\#}$ be the solution to 
   		\begin{equation}
   			\text{argmin}\{\| L\bfx - \bfy\|_{2}:\ \|\bfx\|_{0} 
\leq s \}
   			\label{eq:CP_eq3}
   		 \end{equation}
   		 obtained after $m = O(\log(n))$ iterations of {\tt SubspacePursuit}.
   \State (3) Let $U^{\#} = \{i: \ x_{i}^{\#} < -R\}$ and   $W^{\#} = \{ i: \ x_{i}^{\#} > R \}$.
   \State {\bfseries Output: }  $C_{a}^{\#} = \left(\Omega\setminus W^{\#}\right)\cup U^{\#}$.
\end{algorithmic}
\end{algorithm}

\begin{remark}
{\tt ClusterPursuit} requires as an input an estimate of $|\Omega\bigtriangleup C_a|$, which might not always be available. This is less of an issue than it might first appear as: 
\end{remark}

\begin{enumerate}
	\item Theorem \ref{thm:SCP_ClusterPursuitWorks} will show that as long as $|\Omega\bigtriangleup C_a| \leq s \leq  0.13|C_a|$ {\tt ClusterPursuit} works well. 
	\item If no knowledge of $|\Omega\bigtriangleup C_a|$ is available, one may run {\tt ClusterPursuit} for various values of $s$ and keep the returned cluster with lowest conductance. 
	\item Alternatively, one could consider the Lasso form of problem \eqref{eq:CP_eq3}: 
\begin{equation}
\label{lasso}
	\text{argmin} \left\{ \| L\bfx - \bfy\|_{2} + \lambda\|\bfx\|_1 \right\}= 
	\text{argmin} \left\{ \| L\bfx - \bfy\|_{2} + \lambda\|\bfx\|_0 \right\}
\end{equation}
as the sparse solution is the cluster indicator $\bfone_U-\bfone_W$ which satisfies  $\|\bfx\|_0=\|\bfx\|_1$. We do not analyze this further here. 
\end{enumerate}

\begin{theorem}
$\bfone_{W} - \bfone_{U}$ is the unique solution to Problem \eqref{eq:UnperturbedSparseRecovery}, for any graph $G$ with clusters $C_1,\ldots, C_k$, as long as $|C_a\bigtriangleup \Omega| \leq s < n_1/2$.
\label{thm:Unique_Sol_Unperturbed}
\end{theorem}

\begin{proof}
One can easily verify that $\bfone_{W} - \bfone_{U}$ is {\em a} solution to \eqref{eq:UnperturbedSparseRecovery}, thus it remains to show that it is the unique one. So, suppose that $\bfv$ satisfies $L^{\text{in}}\bfv = \bfy^{\text{in}}$ and that $\bfv \neq \bfone_{W} - \bfone_{U}$. Because $\bfy^{\text{in}} = L^{\text{in}}\bfone_{\Omega}$:
\begin{align*}
& L^{\text{in}}\bfv - L^{\text{in}}\bfone_{\Omega} = 0 \implies \bfv - \bfone_{\Omega} \in \text{ker}(L^{\text{in}}) \implies \bfv - \bfone_{\Omega} = \sum_{b = 1}^{k}\alpha_{b}\bfone_{C_b} \quad \text{( by Theorem \ref{thm:EvecFormBasis})} \\
\implies & \bfv = \sum_{b=1}^{k} \alpha_{b}\bfone_{C_{b}\setminus\Omega} + \sum_{b=1}^{k}(\alpha_{b} + 1)\bfone_{C_{b}\cap\Omega}
\end{align*}
Now if $\alpha_{a} = -1$ and $\alpha_{b} =  0$ for all $b\neq a$ then $\bfv = \bfone_{W} - \bfone_{U}$, which we are assuming is not the case. Hence either $\alpha_{a} \neq -1$, in which case $\|\bfv\|_{0} \geq |C_a\cap\Omega| \geq |C_a| - |C_a\bigtriangleup\Omega| > n_a/2$, or $\alpha_{b} \neq 0$ for $b \neq a$ in which case $\|\bfv\|_{0} \geq |C_{b}\setminus\Omega| \geq |C_b| - |C_a\bigtriangleup\Omega| > n_b/2$ as we are assuming that $|C_a\bigtriangleup\Omega| < n_1/2$ and $n_1 = \min_{b} n_b$. By assumption, $s <  n_1/2$ hence in either case $\bfv$ is infeasible for Problem \eqref{eq:UnperturbedSparseRecovery}, as it does not satisfy the constraint $\|\bfv\|_{0} \leq s$.
\end{proof}

Henceforth, we shall focus on recovering the smallest cluster, $C_1$. We do this to avoid a technical complication in the estimation of $\delta_{\gamma n_a}(L)$ for $a > 1$ (see Theorem \ref{thm:BoundAllQuantities} and Remark \ref{remark:Explaining_why_smallest_cluster}). We note that as long as $n_a \approx n_1$ this is not really an issue, and the proof of Theorem \ref{thm:SCP_ClusterPursuitWorks} will extend to this case, albeit with a tighter bound on $\epsilon$.\\

 Let us now quantify the size of the perturbation in moving from \eqref{eq:UnperturbedSparseRecovery} to \eqref{eq:PerturbedSparseRecovery}. Define $M := L - L^{\text{in}}$ and $\bfe := \bfy - \bfy^{\text{in}}$.  Recall 
from Theorem \ref{thm:PerturbedSP}, that the three key parameters in perturbed compressive 
sensing are the restricted isometry constant of $L$ and:
\begin{equation}
\epsilon_{\bfy} = \frac{\|\bfe\|_{2}}{\|\bfy^{\text{in}}\|_{2}} \quad \text{ 
and } \quad \epsilon^{s}_{L} = \frac{\|M\|_{2}^{(s)}}{\|L^{\text{in}}
\|_{2}^{(s)}}
\label{eq:BoundEpsilons}
\end{equation}
as well as two secondary quantities, $\rho$ and $\tau$. We prove the following:
\begin{theorem}
Suppose that $\mathcal{G}_n$ satisfies (A1)--(A4) and that $|\Omega\bigtriangleup C_1|\leq 0.13n_1$. Then for any $\gamma \in (0,1)$ the following hold almost surely:
\label{thm:BoundAllQuantities}
\end{theorem}

\begin{enumerate}
\item $\epsilon_{\bfy} = o(1)$ and $\epsilon^{ \gamma n_1}_{L} = o(1)$. 
\item $\delta_{\gamma n_1}(L) \leq \gamma + o(1)$. 
\item If $\delta_{3s}(L)\leq 0.45$ then $\rho \leq 0.8751$ and $\tau \leq 55.8490$ for any $s\in (0,n_1/3)$.
\end{enumerate}

\begin{proof}
Part (3) follows by direct computation. For parts (1) and (2) see Appendix~\ref{A:Part2Proof}. 
\end{proof}

We now prove the main result of this section:

\begin{theorem}
\label{thm:SCP_ClusterPursuitWorks}
Suppose that $A$ is the adjacency matrix of $G\sim \mathcal{G}_{n}$ satisfying assumptions (A1)--(A4), and that $\Omega$ satisfies $|C_1\bigtriangleup\Omega| = \epsilon n_1$ with $\epsilon \leq 0.13$. If $C^{\#}$ is the output of {\tt ClusterPursuit} when given inputs $A$, $\Omega$, $\epsilon n_1 \leq s \leq 0.13n_1$ and $R = 0.5$ then:
$$
\frac{\left|C_1\bigtriangleup C_1^{\#}\right|}{|C_1|} = o(1) \quad \text{ almost surely.}
$$
\end{theorem}

\begin{remark}
$s \leq 0.13 n_1$ is a conservative upper bound on $s$ for which the guarantees from \S \ref{subsection:CS} will hold. If one has no further information on $|C_1\bigtriangleup \Omega|$ we recommend using this as the default value of $s$. Empirically (see \S \ref{sec:SyntheticExperiments}) we still observe excellent performance when $s > 0.13n_1$.  
\end{remark}

\begin{proof}
Recall $\bfx^{\#}$ is the solution obtained by $m = O(\log(n))$ iterations of 
{\tt SubspacePursuit} on Problem \ref{eq:PerturbedSparseRecovery}, which we are regarding as a perturbation of 
Problem \ref{eq:UnperturbedSparseRecovery}. Clearly, $0.13n_1 < n_1/2$, hence by Theorem 3.2 $\bfone_{W} - \bfone_{U}$ is the unique solution to \ref{eq:UnperturbedSparseRecovery}. By Theorem \ref{thm:BoundAllQuantities} part (2) we get $\delta_{s}(L) \leq 0.13 + o(1) < 0.15$ almost surely, for large enough $n_1$. Similarly 
$\delta_{3s}(L) \leq 0.45$, again almost surely for $n_1$ large enough. It follows from 
Theorem~\ref{thm:BoundAllQuantities} part (3) that $\rho \leq 0.8751$ and $\tau \leq 55.8490$.  We now appeal to Theorem \ref{thm:PerturbedSP} to obtain:
\begin{equation*}
\frac{\|(\bfone_{U} - \bfone_{W}) - \bfx^{\#}\|_{2}}{\|\bfone_{U} - \bfone_{W}\|_{2}} \leq \rho^{m} + \tau\frac{\sqrt{1 + \delta_{s}}}{1 - 
\epsilon^{s}_{\Phi}}(\epsilon^{s}_{\Phi} + \epsilon_{\bfy}).  
\end{equation*}
The second term on the right-hand side is $o(1)$ by Theorem \ref{thm:BoundAllQuantities}. As 
long as $m \geq \log_{\rho}(1/n) = O(\log(n))$, we obtain that $\rho^{m} = 1/n = o(1)$ too. Thus:
\begin{equation}
\frac{\|(\bfone_{U} - \bfone_{W}) - \bfx^{\#}\|_{2}}{\|\bfone_{U} - \bfone_{W}\|_{2}}\leq o(1) \Longrightarrow \|(\bfone_{U} - \bfone_{W}) - \bfx^{\#}\|_{2} \leq o\left(\|\bfone_{U} - \bfone_{W}\|_{2}\right) = o(\sqrt{n_1})
\label{eq:Sick_of_This}
\end{equation}
As $\|\bfone_{U} - \bfone_{W}\|_2 = \sqrt{|U| + |W|} = \sqrt{\epsilon n_1}$. In Lemma \ref{lemma:Support_Error} below we show that, because $\bfone_{U} - \bfone_{W}$ is a difference of binary vectors, equation \eqref{eq:Sick_of_This} implies that $|U\bigtriangleup U^{\#}| = o(n_1)$ and $|W\bigtriangleup W^{\#}| = o(n_1)$, and hence $|C_1\bigtriangleup C_1^{\#}| = o(n_1)$, as required. 
\end{proof}

\begin{lemma}
\label{lemma:Support_Error}
Consider disjoint $T_1,T_2\subset [n]$ and any $\bfv\in\mathbb{R}^{n}$.  Define $T_1^{\#} = \{i: \ v_i > 0.5\}$ and $T_2^{\#} = \{i: \ v_i < -0.5\}$. If $\|\left(\bfone_{T_1} - \bfone_{T_2}\right) - \bfv\|_2 \leq D$ then:
$$
|T_1\bigtriangleup T_1^{\#}| +  |T_2\bigtriangleup T_2^{\#}| \leq 4D^2. 
$$
\end{lemma}

\begin{proof}
Let $T_{3} := [n]\setminus(T_1\cup T_2)$ and write $\bfv = \bfv^{(1)} + \bfv^{(2)} + \bfv^{(3)}$ where $\bfv^{(i)}$ denotes the part of $\bfv$ supported on $T_{i}$. Observe that:
$$
D^{2} \geq \|(\bfone_{T_1} - \bfone_{T_2}) - \bfv\|_2^{2}  = \|\bfone_{T_1} - \bfv^{(1)}\|_2^{2} + \|-\bfone_{T_2} + \bfv^{(2)}\|_2^{2} + \|\bfv^{(3)}\|^{2}
$$

One can easily verify that:
$$
\|\bfone_{T_1} - \bfv^{(1)}\|_2^{2} = \|\bfone_{T_1\cap T_1^{\#}} - \bfv^{(1)}|_{T_1\cap T_1^{\#}}\|^2 + \| \bfone_{T_1\setminus T_1^{\#}} - \bfv^{(1)}_{T_1\setminus T_1^{\#}}\|_{2}^{2} \geq (0.5)^2 |T_1\setminus T_1^{\#}|
$$ 

Similarly, $\|-\bfone_{T_2} + \bfv^{(2)}\|_2^{2} \geq (0.5)^2|T_2\setminus T_2^{\#}|$, and:
$$
\|\bfv^{(3)}\|_{2}^{2} \geq  \|\bfv^{(3)}|_{T_1^{\#}\setminus T_1}\|^{2} + \|\bfv^{(3)}|_{T_2^{\#}\setminus T_2}\|^{2} \geq (0.5)^2|T_1^{\#}\setminus T_1| + (0.5)^2|T_2^{\#}\setminus T_2|
$$

Putting this all together we get that:
$$
D^2 \geq (0.5)^{2} \left( |T_1\setminus T_1^{\#}| +  |T_2\setminus T_2^{\#}| + |T_1^{\#}\setminus T_1| + |T_2^{\#}\setminus T_2|   \right)  = 0.25 \left(|T_1\bigtriangleup T_1^{\#}| + |T_2\bigtriangleup T_2^{\#}|\right)
$$
\end{proof}

\section{The {\tt RWThresh} algorithm}
\label{sec:RWThresh}
Here, we introduce a simple, diffusion-based local clustering algorithm which we call {\tt RWThresh} (see Algorithm \ref{alg:RWThresh}). We note that {\tt RWThresh} is somewhat similar to other more sophisticated diffusion-based local clustering algorithms, such as {\tt PPR-Grow}, {\tt HK-Grow} and {\tt CapacityReleasingDiffusion}. We do not claim that {\tt RWThresh} outperforms similar existing algorithms; its main utility lies in the fact that it reliably (and provably) produces approximate cuts, $\Omega$, that are of a high enough quality to be used as an initialization for {\tt ClusterPursuit}. 

\begin{algorithm}[h!]
   \caption{ {\tt RWThresh}}
   \label{alg:RWThresh}
\begin{algorithmic}
   \State {\bfseries Input:} Adjacency matrix $A$, a thresholding parameter 
$\epsilon \in (0,1)$, seed vertices $\Gamma\subset C_1$,$\hat{n}_1 \approx n_1$ and 
depth of random walk $t$.
   \State (1) Compute $P = AD^{-1}$ and let $\bfv^{(0)} = D\bfone_{\Gamma}$.
   \State (2) Compute $\bfv^{(t)} = P^{t}\bfv^{(0)}$
   \State (3) Define $\Omega = \tilde{\mathcal{L}}_{(1+\epsilon)\hat{n}_1}(\bfv^{(t)})$.
   \State {\bfseries Output: } $\Omega = \Omega\cup\Gamma$.
\end{algorithmic}
\end{algorithm}

Here $\tilde{\mathcal{L}}_{t}(\cdot)$ is a thresholding operator, similar to $\mathcal{L}_{t}(\cdot)$, but that returns the indices of the $t$ largest, not largest-in-magnitude, components of a vector. To motivate {\tt RWThresh} we observe the following. If $\mathcal{G}_n$ satisfies assumptions (A1)--(A4) then: 
\begin{enumerate}
\item $G_{C_1}$ is sufficiently densely connected that after $t$ steps the random walk has a 
fairly large probability of visiting every $i\in C_1$. 
\item $C_1$ is sufficiently weakly connected to $V\setminus C_1$ that the probability of the 
random walk leaving $C_1$ after $t$ steps is fairly small. 
\end{enumerate} 
Hence Algorithm, \ref{alg:RWThresh} which  
runs a short random walk starting on $\Gamma$ and takes $\Omega$ to be the set of 
vertices most likely to be visited, should produce an $\Omega$ which  
is close to our intuitive notion of a good cluster. Let us quantify this as Theorem \ref{thm:FindingOmegaSSBM}. 

\begin{theorem}
Let $G\sim\mathcal{G}_{n}$ satisfy Assumptions (A1)--(A4) and let $A$ denote the adjacency matrix of $G$. Let $\Omega$ denote the output of {\tt RWThresh} with inputs $A$, any $\epsilon\in (0,1)$, any $t = O(1)$, $\hat{n}_1 = n_1$ and $\Gamma\subset C_1$ with $|\Gamma| = g\epsilon_3^{2t-1}n_1$ for any 
constant $g\in (0,1)$, where $\epsilon_3$ is as in Assumption (A4)). Then $|\Omega\bigtriangleup C_1| \leq (\epsilon + o(1)) n_1$ almost surely. 
\label{thm:FindingOmegaSSBM}
\end{theorem}

\begin{proof}
The proof is left to Appendix~\ref{A:Part1Proof}.
\end{proof}

We note that there are many local clustering algorithms, for example the {\tt PPR-Grow} and {\tt CapacityReleasingDiffusion} algorithms discussed in \S \ref{sec:LiteratureReview}, that require only $|\Gamma| = O(1)$. However, these algorithms tend to return small clusters, typically of size $|C| = O(1)$.  If $\epsilon_3 = O(1/\log(n))$, as it is in the numerical experiments of \S \ref{sec:SyntheticExperiments}, then Theorem \ref{thm:FindingOmegaSSBM} requires that $|\Gamma| = O(n_1/\text{polylog}(n_1))$, which seems to be a reasonable assumption when finding a cluster of size $O(n)$.  In practice, we find it suffices to take $|\Gamma| = 0.01n_1$ or $|\Gamma| = 0.02n_1$.

\section{Using {\tt ClusterPursuit} for local clustering}
\label{sec:CPRWT}
As mentioned earlier, using {\tt RWThresh} to quickly generate a rough approximation to $C_1$, namely $\Omega$, and then using {\tt ClusterPursuit} to then refine this cut leads to a (weakly) local clustering algorithm. Here we verify that this approach, presented below as algorithm \ref{alg:CP+RWT}, works well for our model of graph.

\begin{algorithm}
   \caption{ {\tt CP+RWT}}
   \label{alg:CPRWT}
\begin{algorithmic}
   \State {\bfseries Input:} Adjacency matrix $A$ and seed vertices $\Gamma\subset C_1$. Parameters $\epsilon\in (0,0.13)$, $s \approx \epsilon n_1$, $R\in [0,1)$, $\hat{n}_1\approx n_1$, $t\in\mathbb{Z}_{+}$  
   \State (1) Let $\Omega = \text{\tt RWThresh}(A,\epsilon,\Gamma,\hat{n}_1,t)$ 
   \State (2) Let $C_1^{\#} = \text{\tt ClusterPursuit}(A,s,R)$
   \State {\bfseries Output: } $C_1^{\#}$
\end{algorithmic}
\label{alg:CP+RWT}
\end{algorithm}

\begin{theorem}
\label{thm:CP+RWT_Works_Well}
Let $G\sim\mathcal{G}_n$ satisfy Assumptions (A1)--(A4) and let $A$ denote the adjacency matrix of $G$. Let $C_{1}^{\#}$ denote the output of {\tt CP+RWT} with inputs $A$, $\epsilon\in (0,0.13)$, $R = 0.5$, $\hat{n}_1 = n_1$, any $t = O(1)$, any $s$ satisfying $\epsilon < s \leq 0.13n_1$ and $\Gamma \subset C_1$  with $|\Gamma| = g\epsilon_{3}^{2t-1}n_1$ for any constant $g\in (0,1)$, where $\epsilon_3$ is as in Assumption (A4). Then:
$$
\frac{\left| C_1\bigtriangleup C_1^{\#}\right|}{|C_1|} = o(1)
$$ 
almost surely, for large enough $n_1$. 
\end{theorem}

\begin{proof}
By Theorem \ref{thm:FindingOmegaSSBM}, the call to {\tt RWThresh} in Step (1) of {\tt CP+RWT} almost surely returns an $\Omega$ satisfying $|\Omega\bigtriangleup C_1| \leq (\epsilon + o(1))n_1$ for input parameters with the given values. For large enough $n_1$, we have that $(\epsilon + o(1))n_1\leq  s \leq 0.13n_1$, hence the call to {\tt ClusterPursuit} in Step (2) of {\tt CP+RWT} returns $C_{1}^{\#}$ with $\left| C_1\bigtriangleup C_{1}^{\#}\right|/|C_1| = o(1)$ by Theorem \ref{thm:SCP_ClusterPursuitWorks}, again almost surely. 
\end{proof}

\begin{remark}
In practice (see \S \ref{sec:NumericalExperiments}) we find it generally suffices to take $t = 3$. If $C_1$ is densely connected, one might consider a smaller value of $t$, and conversely one might choose a larger value (say $t = 5$) if $C_1$ is sparsely connected.
\end{remark}

\section{Using {\tt ClusterPursuit} for semi-supervised clustering}
\label{sec:ICPRWT}
In the (global) semi-supervised clustering problem, one is given a small set of seed vertices $\Gamma_a\subset C_a$  in each cluster, usually referred to in this context as ``labeled data''. The goal here is to find a partition into disjoint sets: $V = C_{1}^{\#}\cup C_{2}^{\#}\cup\ldots\cup C_{k}^{\#}$ that closely resembles the ground truth partition $V = C_1\cup C_2\cup\ldots \cup C_k$. An iterated version of {\tt CP+RWT}, which we call {\tt ICP+RWT}, can be used to solve this problem. {\tt ICP+RWT} is presented as algorithm \ref{alg:ICP+RWT}. Note that in the second line of the for loop we use the shorthand $G^{(a+1)} = G^{(a)}\setminus C^{\#}_{a}$ to denote the graph formed from $G^{(a)}$ by removing the vertices $C_{a}^{\#}$. We do not analyze the theoretical performance of {\tt ICP+RWT} here\footnote{There is a minor technical difficulty: one needs to show that if $G$ is drawn from a model satisfying assumptions (A1)--(A4) then each $G^{(a)}$ is also drawn from a model satisfying assumptions (A1)--(A4).} but we provide numerical evidence that {\tt ICP+RWT} is competitive with state-of-the-art semi-supervised graph clustering algorithms in \S \ref{sec:MLBenchmarks}.

\begin{algorithm}
   \caption{ {\tt ICP+RWT}}
   \label{alg:ICPRWT}
\begin{algorithmic}
   \State {\bfseries Input:} Adjacency matrix $A$, labeled data $\Gamma_{a}\subset C_a$ for $a= 1,\ldots, k$. Parameters $\epsilon\in (0,1)$, $R\in [0,1)$, $\hat{n}_a\approx n_a$ and $s_a \approx \epsilon n_a$ for $a = 1,\ldots,k$,  and $t\in\mathbb{Z}_{+}$
   \State {\bfseries Initialize}: $G^{(1)} = G$ and $A^{(1)} = A$. 
   \For{$a=1,\ldots k$}
   		\State Let $C_{a}^{\#} = \text{\tt CP+RWT}(A^{(a)},\Gamma_{a},\epsilon,R,s_a, \hat{n}_a,t)$  
   		\State Let $G^{(a+1)} = G^{(a)}\setminus C^{\#}_{a}$ and let $A^{(a+1)}$ be the adjacency matrix of $G^{(a+1)}$.  
   	\EndFor
   \State {\bfseries Output: } $C_1^{\#},\ldots, C_{k}^{\#}$
\end{algorithmic}
\label{alg:ICP+RWT}
\end{algorithm}

\section{Computational Complexity}
\label{sec:Complexity}
In this section we discuss the run times of the algorithms introduced in this paper. Let $\mathcal{T}_{m}$ denote the cost of a matrix-vector multiply with $A$, $L$ or $P$ (they are all of the same magnitude).

\begin{theorem}
\label{thm:Run_Time:RWT}
{\tt RWThresh} requires $O(n\log(n) + t\mathcal{T}_m)$ operations, where $t$ is the depth of the random walk.
\end{theorem}

\begin{proof}
Computing $\bfv^{(t)}$ requires $t$ matrix-vector multiplies and hence requires $O(t\mathcal{T}_m)$ operations. Sorting $\bfv^{(t)}$ in order to find $\Omega$ requires $O(n\log(n))$ operations.
\end{proof}

Let us now analyze the complexity of {\tt ClusterPursuit}

\begin{theorem}
\label{thm:Run_Time:CP}
{\tt ClusterPursuit} requires $O\left(\mathcal{T}_m\log(n)\right)$ operations.
\end{theorem}

\begin{remark}
 Note that if $A$ is stored as a sparse matrix then $\mathcal{T}_m = O(nd_{\max})$ in which case  the run time of {\tt ClusterPursuit} becomes $O(nd_{\max}\log(n))$. 
\end{remark}

\begin{proof}

The run time of {\tt ClusterPursuit} is dominated by the cost of the call to {\tt SubspacePursuit} (see Algorithm \ref{algorithm:SP}) in step (3) which costs $m$ times the cost of each iteration. We now bound the cost of each iteration. The cost of the $j$-th iteration is dominated by the cost of solving the least squares problem:
$$
\argmin_{\bfz\in\mathbb{R}^{n}} \left\{\left\| L\bfz - \bfy\right\|_2: \ \text{supp}(\bfx)\subset \hat{S}^{j} \right\}.
$$
(step (4) in the ``for'' loop of Algorithm \ref{algorithm:SP}). Because of the support condition, and because $|\hat{S}^{j}| = 2s\leq 0.26 n_1$, this is equivalent to the least squares problem:
\begin{equation}
\argmin_{\bfz\in\mathbb{R}^{2s}} \left\{\left\| L_{\hat{S}^{j}}\bfz - \bfy\right\|_2\right\}
\label{eq:Needell_LSQR}
\end{equation}
We recommend using an iterative method, such as conjugate gradient (in our implementation we use MATLAB's {\tt lsqr} operation). Fortunately, as pointed out in \cite{NT09}, the matrix in question, $L_{\hat{S}^{j}}$ is extremely well conditioned. This is because $\delta_{2s}(L) \leq \delta_{3s}(L) \leq 0.45$, as shown in the proof of Theorem \ref{thm:SCP_ClusterPursuitWorks}. By \cite{NT09}, specifically Proposition 3.1 and the discussion of \S 5, this implies that the condition number is small:
\begin{equation*}
\kappa(L_{\hat{S}^j}^{\top}L_{\hat{S}_j}) := \frac{\lambda_{\max}(L_{\hat{S}^j}^{\top}L_{\hat{S}_j})}{\lambda_{\min}(L_{\hat{S}^j}^{\top}L_{\hat{S}_j})} \leq \frac{1+\delta_{2s}}{1-\delta_{2s}} \leq 2.64
\end{equation*}
The upshot of this is that it only requires a constant number of iterations of conjugate gradient to approximate the solution to the least-squares Problem \eqref{eq:Needell_LSQR} to within an acceptable tolerance. Indeed, Corollary 5.3 of \cite{NT09} argues that three iterations suffices. We play it safe by performing ten iterations. The cost of each iteration of conjugate gradient is equal to (a constant times) the cost of a matrix vector multiply by $L_{\hat{S}_{j}}$ or $L^{\top}_{\hat{S}_{j}}$, which is $\mathcal{T}_{m}$. Hence the total cost of step (3) of {\tt ClusterPursuit} is $O(m\mathcal{T}_{m}) = O(\log(n)\mathcal{T}_{m})$ because we are taking $m=O(\log(n))$. 
\end{proof}

As a direct consequence of Theorems \ref{thm:Run_Time:RWT} and \ref{thm:Run_Time:CP}, we get that {\tt CP+RWT} runs in time $O((nd_{\max}\log(n))$. If the number of clusters, $k$, is $O(1)$, we get that {\tt ICP+RWT} also runs in time $O((nd_{\max}\log(n))$.

\section{Comparison with Existing Literature}
\label{sec:LiteratureReview}
{\tt ClusterPursuit} can naturally be compared with other cut improvement algorithms such as {\tt FlowImprove} \cite{A08}, {\tt LocalFlow} \cite{O14} and {\tt SimpleLocal} \cite{V16}. We note that the performance guarantees for these three algorithms are of a different flavor to ours. Specifically, and translating into the notation of this paper, they bound the conductance of the improved cut, $C_1^{\#}$, by some function of the original cut, $\Omega$. In contrast, our performance guarantees for {\tt ClusterPursuit} are of a more statistical nature. In terms of run-time, {\tt LocalFlow} and {\tt SimpleLocal} are strongly local, so have run times $O(\text{vol}(\Omega)^{\alpha})$ for $\alpha \geq 1$. While this is certainly better than {\tt ClusterPursuit} for finding small clusters, {\em ie} when $|\Omega| = O(1)$, these run times become less attractive for even moderate sized clusters, {\em eg} $|C_1| = O(\sqrt{n})$. In \S \ref{sec:NumericalExperiments} we demonstrate that {\tt ClusterPursuit} is several orders of magnitude faster than {\tt FlowImprove} and {\tt SimpleLocal} in the regime $|C_1| = O(n)$. \\

The idea of combining a fast, diffusion based clustering algorithm with a refinement procedure to create a local clustering algorithm is not new. See, for example, the algorithms  {\tt LEMON} \cite{HSBH15,LHKBH18}, {\tt LOSP} and {\tt LOSP++} \cite{LHBH15}, {\tt LBSA} \cite{SHBH19},  and {\tt FlowSeed} \cite{V19}. We compare {\tt CP+RWT} to a selection of these algorithm in \S \ref{sec:NumericalExperiments}. We note that there exist many diffusion-based local clustering algorithms that may find better approximations to $C_1$ than {\tt RWThresh}. See for example, {\tt PPR-Grow} \cite{ACL07}, {\tt HK-Grow} \cite{KG14} or {\tt CapacityReleasingDiffusion} \cite{W17}. We emphasize that the main advantage of {\tt RWThresh} is that it rapidly and provably finds good enough initial cuts, $\Omega$, to be fed into {\tt ClusterPursuit}. We show in \S \ref{sec:NumericalExperiments} that the combination {\tt CP+RWT} typically outperforms these diffusion-only approaches, particularly for large, sparsely connected clusters. \\

The analysis of {\tt CP+RWT} contained in \S \ref{sec:ClusterPursuit}--\ref{sec:CPRWT} can be compared to the recent works \cite{W17} and \cite{H19}. In both the performance of a local clustering algorithm on graphs drawn from a certain probabilistic model is studied. In both papers, the model is more general in one sense: there is no restriction on the structure of $V\setminus C_1$, but more restrictive in other senses: the ratio $d^{\text{out}}/ d^{\text{in}}$ must be at most $O(1/ \log^2(n_1)) $ in \cite{W17} while the results in \cite{H19} are most meaningful when $n_1 = O(1)$ and $d^{\text{out}} = O(1)$. In contrast, our results tackle the regime where $n_1 = O(n)$ and $d^{\text{out}}/ d^{\text{in}}$ can be bounded by an arbitrarily slowly decaying function of $n_1$. \\

Finally, we mention several recent works that combine notions of sparsity and local clustering. In particular, we mention the works of Fountoulakis, Gleich, Mahoney {\em et al} \cite{GM14,F16,H19} which introduce and study the $\ell_1$ regularized page rank problem. The algorithms {\tt LOSP} and {\tt LOSP++} also set up and solve a sparse recovery problem, although with an additional non-negativity requirement. However, to the best of the authors' knowledge, {\tt ClusterPursuit} is the first algorithm that explicitly phrases the problem of improving a cut, $\Omega$, as the problem of finding a sparse change to the indicator vector $\bfone_{\Omega}$. 

\section{Which Probabilistic Models Satisfy our Assumptions?}
\label{sec:SBM}
First, we verify that a well-studied model of graphs with clusters, namely the stochastic block model, satisfies Assumptions (A1)--(A4) of \S \ref{sec:DataModel}. We first remind the reader of the simpler \ER model:
 
\begin{definition}
We say $G = (V,E)$ is drawn from the \ER model on $n$ vertices with parameter $p$ (and write $G \sim \text{ER}(n,p)$) if $V = [n]$ and $\mathbb{P}[\{i,j\} 
\in E] = p$ for $i,j \in V$, with all such probabilities being independent.
\end{definition}  

\begin{definition}[\cite{HL83,A17}]
Let $\mathbf{n} = (n_1,\ldots, n_k)$ be a vector of positive integers, and let $P$ be a $k\times k$ symmetric matrix with entries $P_{ab} 
\in [0,1]$ for all $a,b$. We say a graph $G = (V,E)$ is drawn from the Stochastic Block Model (written $G \sim 
\text{SBM}(\mathbf{n},P)$) if there exists a partition $V = C_1\cup C_2\ldots \cup C_k$ with $|C_a| = n_a$ such that any 
vertices $i\in C_a$ and $j\in C_b$ are connected by an edge with probability $P_{ab}$, and all edges are inserted independently.
\end{definition} 

Note that if $G \sim \text{SBM}(\mathbf{n},P)$ then each $G_{C_a}\sim \text{ER}(n_a, P_{aa})$. Without loss of generality, we shall assume that $n_{1}\leq n_{2} \leq \ldots \leq n_{k}$. In an appendix, we shall prove the following:

\begin{theorem}
Suppose that $n_1\to \infty$, $P_{aa} = \omega\log(n)/n_a$ for any $\omega\to \infty$ and $P_{ab} = (\beta + o(1))\log(n)/n$ for any $a\neq b$ where $\beta$ is a constant. Then $\text{SBM}(\mathbf{n},P)$ satisfies assumptions (A1)--(A4). 
\label{theorem:Assumptions_for_SBM}
\end{theorem}

\begin{proof}
See Appendix \ref{section:Assumptions_for_SBM}.
\end{proof}

As a consequence of this theorem we have that, given a small fraction of vertices in $C_1$, {\tt CP+RWT} will reliably return a $C_1^{\#}$ with $|C_1\bigtriangleup C_1| = o(n_1)$. We experimentally confirm this in \S \ref{sec:NumericalExperiments} for $\omega \sim \log(n)$. In this regime we have that $d_{\max} = O(\log^{2}(n))$ with high probability, hence the run time of {\tt CP+RWT} is $O(n\log^{3}(n))$ by Theorem \ref{thm:Run_Time:CP}.  \\

It is interesting to contrast this result with what is known for the global clustering problem for the stochastic block model. There are several unsupervised algorithms, see for example\cite{AS15} and \cite{MNS18}, that return a partition $V = C_1^{\#}\cup C_{2}^{\#}\cup\ldots\cup C_{k}^{\#}$ such that $C_{a}^{\#} = C_a$ with high probability. However these approaches either have impractically high run times \cite{MNS18} or are tricky to implement in practice \cite{AS15}. In contrast, {\tt CP+RWT} has a low run time, in theory and in practice, and can be implemented in a few lines of code. In addition, the ``one cluster at a time'' nature of {\tt CP+RWT} affords an additional flexibility that may be useful in certain circumstances. \\

On the other hand, we have had less success with using {\tt CP+RWT} for certain random geometric graphs arising as $K$-NN graphs of point clouds in $\mathbb{R}^{d}$. We note that {\tt CP+RWT} is most effective when the adjacency matrix of the $K$-NN graph is sparse but has its non-zero entries uniformly distributed. In contrast, for certain artificial data sets, for example points drawn from a thickened line or sphere embedded in a high dimensional space, this adjacency matrix tends to exhibit a banded structure---at least when nearest neighbors are determined using the Euclidean metric. Experimentally, we have observed that {\tt CP+RWT} performs poorly on these data sets. However, this problem is to a large extent particular to the use of the Euclidean metric. In particular, when a data-driven metric such as those detailed in \cite{MD19} is used to construct the $K$-NN graph, {\tt CP+RWT} performs much better. Moreover, even when using the Euclidean metric {\tt CP+RWT} still performs extremely well on real data sets, such as MNIST, COIL and Optdigits, which are frequently thought of as consisting of data points drawn from a low-dimensional manifold embedded in a high dimensional space (see \S \ref{sec:MLBenchmarks}). 

\section{Numerical Experiments}
\label{sec:NumericalExperiments}

We compare the algorithms {\tt ClusterPursuit}, {\tt CP+RWT} and {\tt ICP+RWT} to the state of the art on the various problems they are designed to solve. Specifically, in \S 10.1 we compare the performance of {\tt ClusterPursuit} on the cut improvement task to two baseline algorithms, namely {\tt FlowImprove} and {\tt SimpleLocal}, for graphs drawn from the stochastic block model. We also compare {\tt CP+RWT} to the local clustering algorithms {\tt HK-Grow}, {\tt PPR-Grow} and {\tt LBSA} for the same data.\footnote{ While there are certainly other worthy local clustering algorithms that deserve to be included, such as {\tt CapacityReleasingDiffusion} \cite{W17} and {\tt FlowSeed} \cite{V19}, we stick to algorithms with a freely available MATLAB implementation}. In \S 10.2 we repeat this experiment for social networks. We take care to choose our data sets and performance measures to allow for easy comparison with similar work in \cite{W17}. In \S \ref{sec:MLBenchmarks} we test the performance of {\tt ICP+RWT} on two data sets commonly studied in the machine learning community---MNIST and OptDigits. We provide a detailed description of the implementation of all algorithms considered in Appendix \ref{sec:NumericalParameters}.

\subsection{Synthetic Data Sets}
\label{sec:SyntheticExperiments}
We consider graphs drawn from $\text{SBM}(\mathbf{n}^{(i)},P^{(i)})$ for two different sets of parameters. The first 
set: $\mathbf{n}^{(1)} = (n_1, 1.5n_1, 2.5n_1, 5n_1)$ and $P^{(1)}$  with $P_{aa} = \log^2(n)/2$ and $P_{ab} = 
5\log(n)/n$ for all $a\neq b$ is designed to satisfy the conditions of Theorem \ref{theorem:Assumptions_for_SBM} while 
presenting a challenge to existing clustering algorithms. The second set: $\mathbf{n}^{(2)} = (n_1,10n_1)$ and $P^{(2)} 
= \left[\begin{smallmatrix} 2\log^2(n)/n & \log(n)/n \\ \log(n)/n & \log(n)/n \end{smallmatrix} \right]$ goes beyond the 
assumptions of Theorem \ref{theorem:Assumptions_for_SBM} and is essentially the planted cluster model studied in 
\cite{H19} and elsewhere. For both sets of parameters we perform two experiments. In the first we test the performance 
of the three cut improvement algorithms when initialized with an $\Omega$ ``close'' to $C_1$. This $\Omega$ is found 
using {\tt RWThresh}. In the second we compare the performance of {\tt CP+RWT} with the performance of the local 
clustering algorithms mentioned above. For both experiments we report both run time and accuracy, 
as measured by the Jaccard Index in Figure~\ref{fig:SCP_Synthetic_Experiments} and in 
Figure~\ref{fig:SCP_Synthetic_Experiments_Weak_Background}, respectively.  

\begin{figure}[htpb]
\centering
 \includegraphics[width =1.8in]{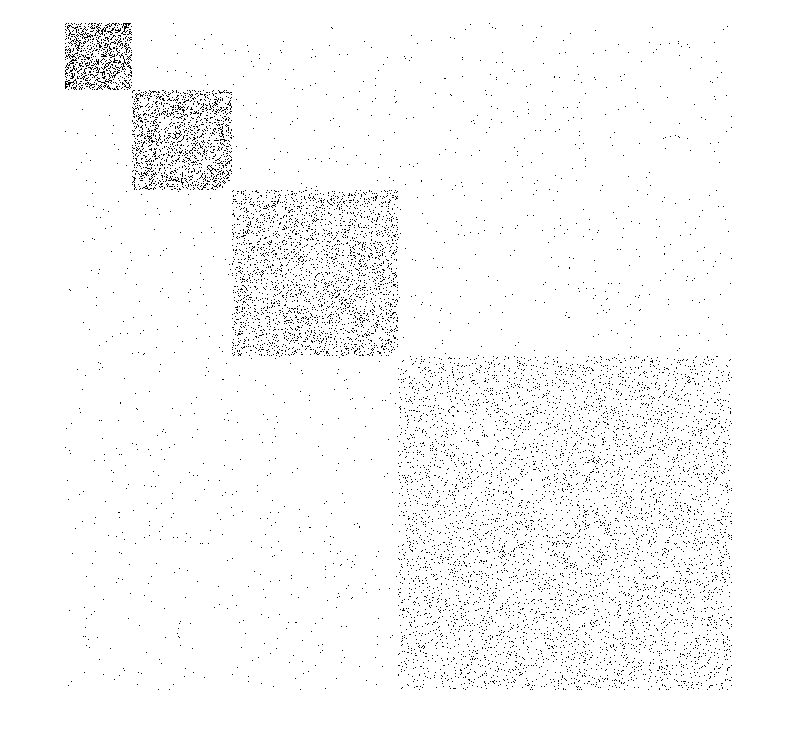} 
\includegraphics[width =1.8in]{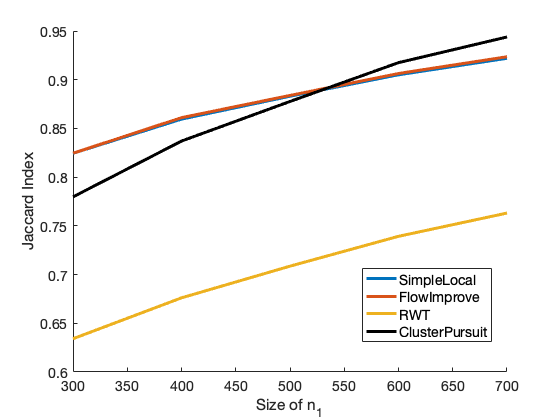}
\includegraphics[width =1.8in]{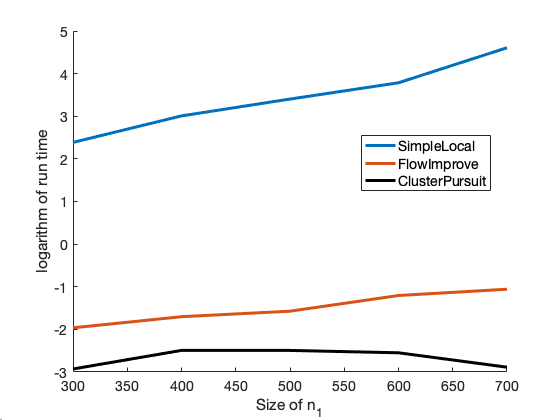} \\

\includegraphics[width =1.8in]{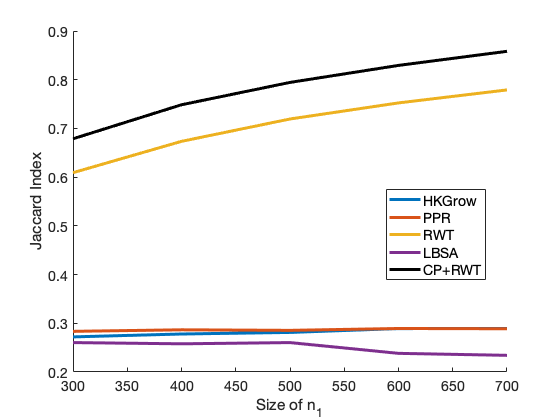}
\includegraphics[width =1.8in]{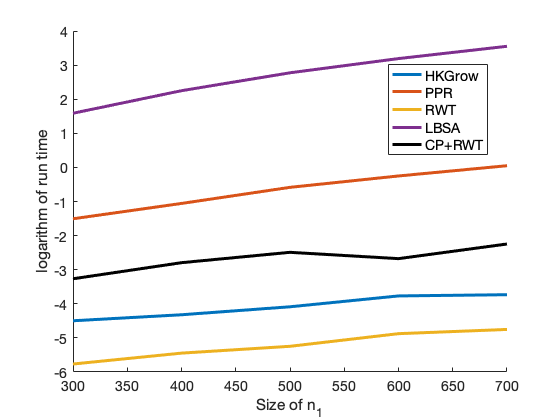}
\includegraphics[width =1.8in]{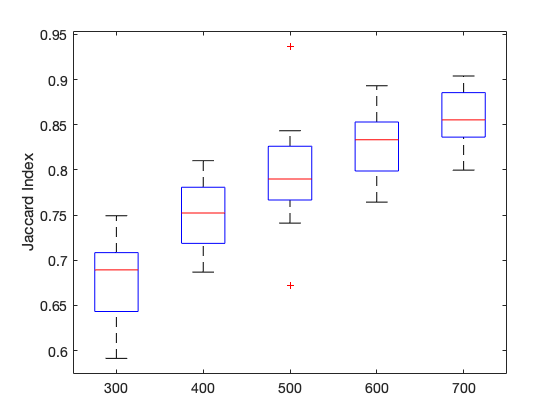}

\caption{ {\em Top row, left to right:} Stylized representation of the adjacency matrix of graphs drawn from 
$\text{SBM}(\mathbf{n}^{(1)},P^{(1)})$, Jaccard index for results of cut improvement ({\tt SimpleLocal} and {\tt 
FlowImprove} always have the same Jaccard index) and(log. of) run time for the three cut improvement algorithms. 
Note that {\tt ClusterPursuit} is at least an order of magnitude faster than the other two, even though {\tt 
FlowImprove} is implemented in C. 
{\em Bottom row, left to right:} Jaccard index for local clustering (The poor performance of the other methods is not an implementation issue. Rather, it is a consequence of the small gap between $P^{(1)}_{aa}$ and $P^{(1)}_{ab}$). (Log. of) run time for local clustering. Box plot of Jaccard index for {\tt CP+RWT}. \label{fig:SCP_Synthetic_Experiments}}
\end{figure}

\begin{figure}[htpb]
\centering
 \includegraphics[width =1.8in]{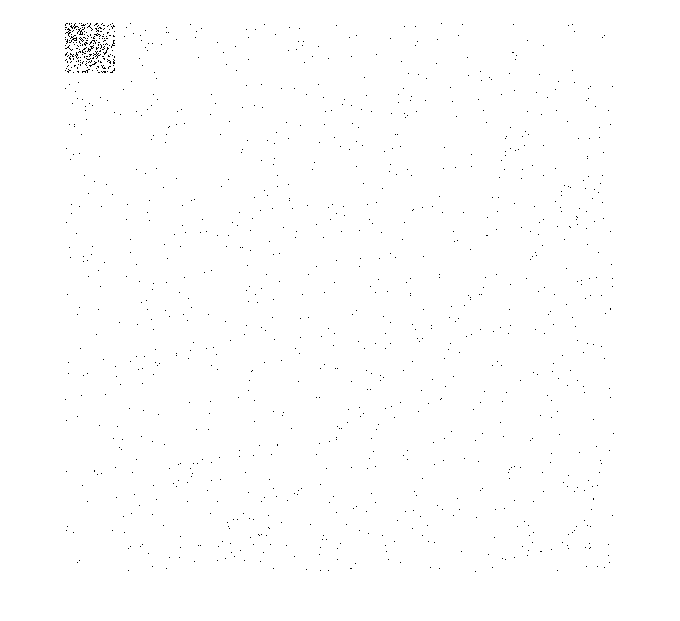} 
\includegraphics[width =1.8in]{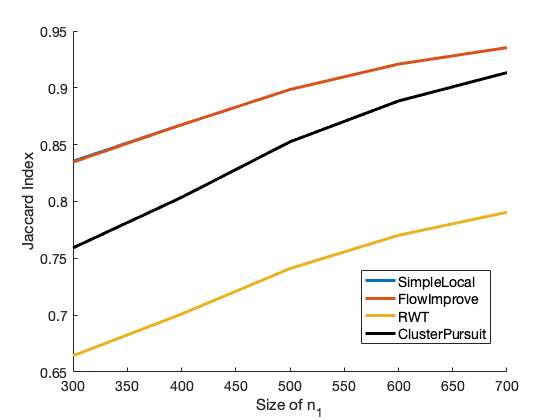}
\includegraphics[width =1.8in]{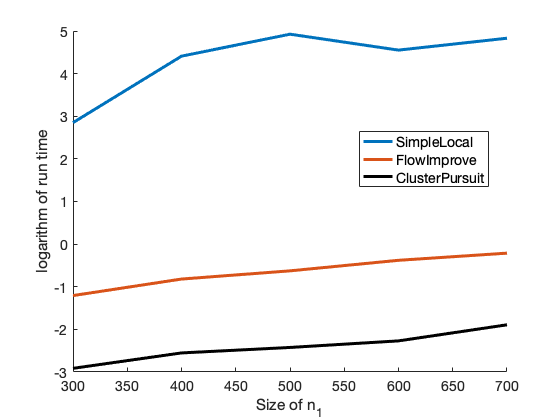} \\

\includegraphics[width =1.8in]{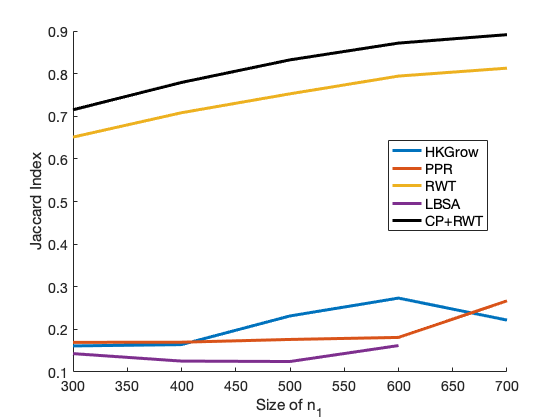}
\includegraphics[width =1.8in]{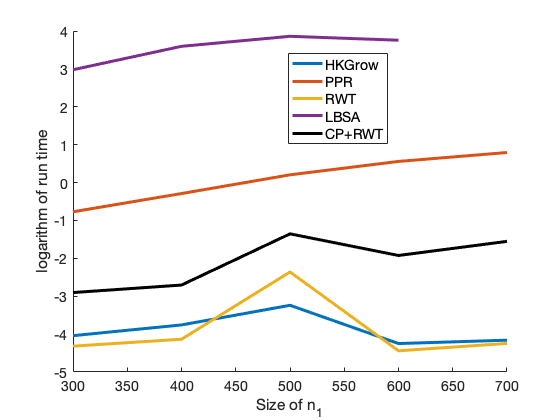}
\includegraphics[width =1.8in]{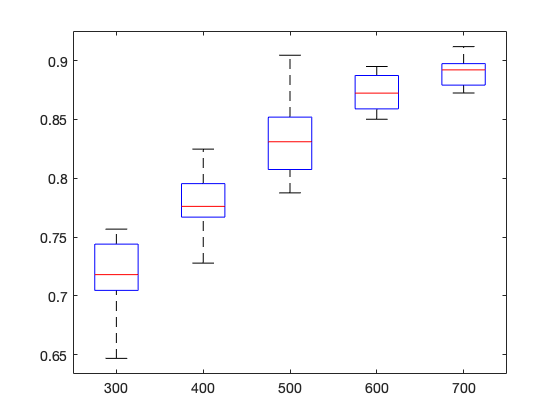}

\caption{ {\em Top row, left to right:} Stylized representation of the adjacency matrix of graphs drawn from $\text{SBM}(\mathbf{n}^{(2)},P^{(2)})$, Jaccard index for results of cut improvement (again, {\tt SimpleLocal} and {\tt FlowImprove} always have the same Jaccard index). (log. of) Run time for the three cut improvement algorithms. {\em Bottom row, left to right:} Jaccard index for local clustering (Again, the poor performance of the benchmark methods is a consequence of the challenging SBM parameters chosen).  (Log. of) run time for local clustering. Box plot of Jaccard index for {\tt CP+RWT}.
\label{fig:SCP_Synthetic_Experiments_Weak_Background} }
\end{figure}

\subsection{Social Networks}
\label{sec:SocialNetworks}
The {\tt facebook100} dataset consists of anonymized Facebook friendship networks at $100$ American universities, and was first introduced and studied in \cite{TMP12}. Certain demographic markers (year of entry, residence etc.) were also collected in an anonymized format. One can think of vertices sharing the same marker as defining a ground truth cluster, although some of these clusters are extremely noisy. We focus on four clusters identified in \cite{W17} as having good ({\em ie} low) or moderately good conductance scores, namely Johns Hopkins class of 2009, Rice University dorm 203, Simmons College class of 2009 and Colgate University class of 2006. The details of these clusters are displayed in Table \ref{table:NetworkStatistics}. For ease of comparison with the results of \cite{W17} we report accuracy using precision and recall scores. We remind the reader that, in the notation of this paper, $\text{precision} = |C_1\cap C_1^{\#}|/|C_1^{\#}|$ and $\text{recall} = |C_1\cap C_1^{\#}|/|C_1|$. It is desirable to have both of these values as close to $1$ as possible. For all four experiments we take $\Gamma$ to be selected uniformly and at random from $C_1$, with $|\Gamma| = 0.02n_1$. We average over fifty independent trials. There results are shown in Figure~\ref{fig:Social_Network_Results}. 

\begin{table}[htpb]
\centering
	\begin{tabular}{ccccc}
	\hline
		School & Cluster & Size of graph & Size of Cluster & Conductance \\ 
		Johns Hopkins & Class of 2009 & $5180$ & $910$ & $0.21$ \\
		Rice & Dorm. 203 & $4087$ & $406$ & $0.47$ \\
		Simmons & Class of 2009 & $1518$ & $289$ & $0.11$ \\
		Colgate & Class of 2006 & $3482$ & $557$ & $0.49$ \\
	\hline
	\end{tabular}
	\caption{Basic properties of the four social networks studied.}
	\label{table:NetworkStatistics}
\end{table}  

\begin{figure}[htpb]
\centering
\includegraphics[width =2in]{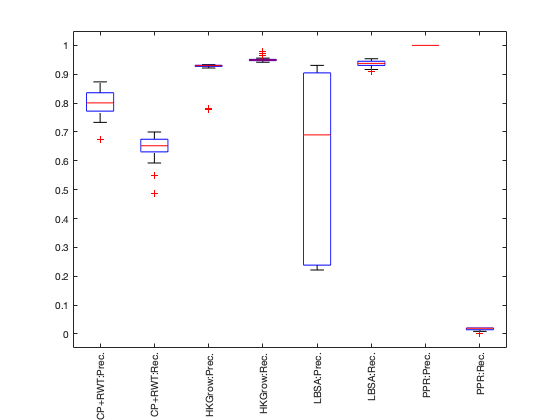}
\includegraphics[width =2in]{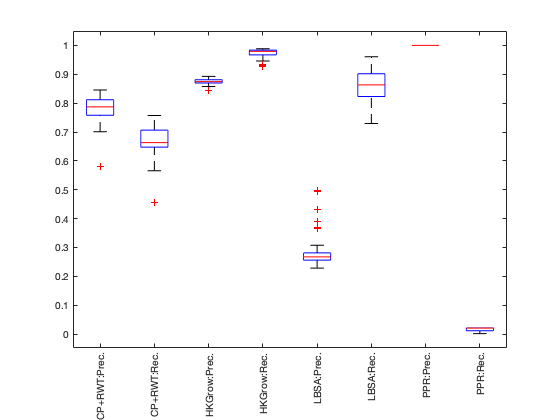} \\

\includegraphics[width =2in]{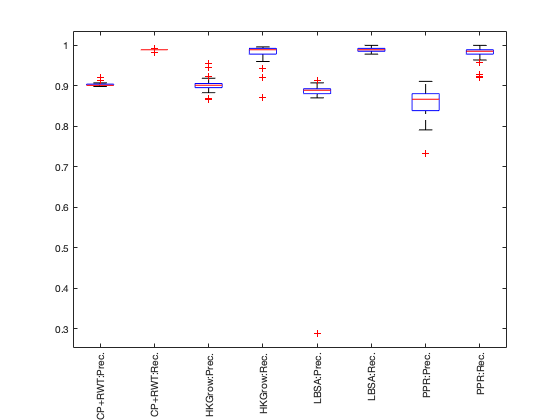}
\includegraphics[width =2in]{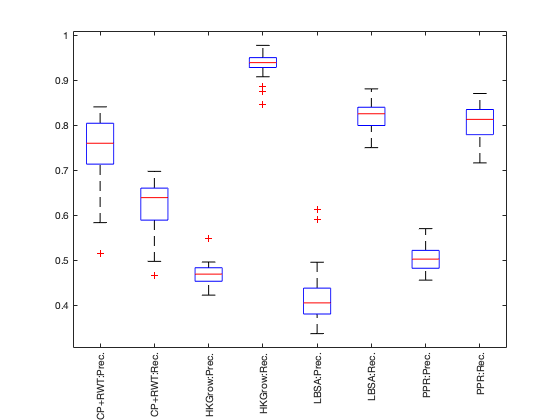}
\caption{Precision and Recall for various local clustering algorithms on the social networks described in Table \ref{table:NetworkStatistics}. Clockwise from top left: Johns Hopkins, Rice, Colgate and Simmons. Note that {\tt CP+RWT} consistently achieves high precision without sacrificing recall.
\label{fig:Social_Network_Results}}
\end{figure}

\subsection{Machine Learning Benchmarks}
\label{sec:MLBenchmarks}
We consider two venerable benchmark data sets:\\

{\bf OptDigits.} This data set consists of grayscale images of handwritten digits 0--9 of size $8\times 8$. There are $n = 5620$ images and the clusters are fairly well balanced with approximately $560$ images of each digit.

{\bf MNIST.} This data set also consists of grayscale images of the handwritten digits 0--9 although here there are $n = 70\ 000$ images, all of size $28 \times 28$. There are approximately $7\ 000$ images of each digit. \\

For each data set we form a $k$-NN graph using the procedure presented in \cite{JME18} and described in detail in Appendix \ref{sec:PreprocessingImages}. The labeled data, $\Gamma_a$, was sampled uniformly at random from $C_a$, and each is of size $g|C_a|$. The accuracy of the classification given by {\tt ICP+RWT}, for increasing $g$, is presented in Table \ref{table:ICPRWT_results}. All results are averaged over twenty independent trials.

\begin{table}
	\centering
	\begin{tabular}{c|ccccc}
	\hline
	\backslashbox{Data Set}{$\%$ Labeled Data} & $0.5$ & $1$ & $1.5$ & $2$ & $2.5$  \\
	\hline
	MNIST & $96.41\%$ & $97.32\%$ & $97.44\%$ & $97.52\% $ & $97.50\%$ \\
	OptDigits & $91.88\%$ & $95.47\%$ & $97.16\%$ & $98.06\%$ & $98.08\%$ \\
	\hline
	\end{tabular}
	\caption{Classification accuracy, as a function of amount of labeled data, for {\tt ICP+RWT} on two well-studied benchmark data sets.}
	\label{table:ICPRWT_results}
\end{table}

\begin{table}[htpb]
\centering
\begin{tabular}{ccc}
\hline
Method & Labeled & Accuracy \\
\hline
TVRF \cite{YT18} & $600$ & $96.8\%$ \\
{\tt ICP+RWT} & $700$ & $97.32\%$ \\
Multi-Class MBO with Auction Dynamics \cite{JME18} & $700$ & $97.43\% $ \\
{\tt ICP+RWT} & $1050$ & $97.44\%$ \\
Ladder Networks \cite{RBHVR15} & $1000$ & $99.16\%$ \\
\hline
\end{tabular}
\caption{Comparing {\tt ICP+RWT} to other, state-of-the-art, semi-supervised methods on MNIST. TVRF and Multi-Class MBO are graph-based, and have similar run times to {\tt ICP+RWT}. The Ladder Network approach uses a deep neural network and hence requires training ($\sim 2$ hours on a GPU) before it can be used for classification.}
\label{table:MNIST_Results}
\end{table}

\bibliographystyle{plain}

\begin{appendices}

\section{Restricted Isometry Property for Laplacians\label{A:Part2Proof}}
In this section, we prove parts (1) and (2) of Theorem \ref{thm:BoundAllQuantities}. We proceed via a series of lemmas.
  
\subsection{Restricted Isometry Property for $L^{\text{in}}$}
 
\begin{lemma}
\label{lemma:GenGraphDeltaBound}
Let $G$ be any connected graph on $n_0$ vertices, and let $s < n_0$. Let $\lambda_{i}:= \lambda_{i}(L)$ denote the $i$-th smallest eigenvalue of $L$. Then:
\begin{equation*}
\delta_{s}(L) \leq \max\{1 - \lambda_{2}^{2}\left(\frac{d_{\min}}{d_{\max}} - \frac{d_{\max}}{d_{\min}}\frac{s}{n_0}\right),
\lambda_{\max}^{2}-1\}.  
\end{equation*}
\end{lemma}

\begin{proof}
Recall that the $s$-Restricted Isometry Constant $\delta_{s}(L)$ is the smallest $\delta$ such that, for any $\bfv$ with $\|\bfv\|_{0} \leq s$ and $\|\bfv\|_{2} = 1$:
$(1-\delta) \leq \|L\bfv\|_{2}^{2} \leq (1+\delta). $
The RHS bound is straightforward since 
\begin{equation*}
\|L\bfv\|_{2} \leq \|L\|_{2}\|\bfv\|_{2} = \lambda_{\max}(1) = \lambda_{\max}.
\end{equation*}
The LHS bound requires some work. Recall that $L = I - D^{-1}A$. This matrix is not symmetric, but $L^{\text{sym}} = I - D^{-1/2}AD^{-1/2}$ is. By Lemma \ref{lemma:Evals_of_L_Lsym_P} $L$ and $L^{\text{sym}}$ have the same eigenvalues. Let $\bfw_1,\ldots, \bfw_{n_0}$ be an orthonormal eigenbasis for $L^{\text{sym}}$. 
These eigenvectors are well studied (see, for example, \cite{C96}) and in particular $\bfw_{1} = 
\frac{1}{\sqrt{\text{vol}(G)}}D^{1/2}\bfone$ where $\bfone$ is the all-ones vector. Observe that:
\begin{equation*}
L\bfv = D^{-1/2}\left(D^{1/2}LD^{-1/2}\right)D^{1/2}\bfv = D^{-1/2}L^{\text{sym}}D^{1/2}\bfv = D^{-1/2}L^{\text{sym}} \bfz, 
\end{equation*}
where $\bfz := D^{1/2}\bfv$. It follows that:
\begin{equation}
\|L\bfv\|_{2}  = \| D^{-1/2}L^{\text{sym}} \bfz\|_{2} \geq \frac{1}{\sqrt{d_{\max}}}\|L^{\text{sym}} \bfz\|_{2}. 
\label{eq:LyBound}
\end{equation}
Express $\bfz$ in terms of the orthonormal basis $\{\bfw_1,\ldots, \bfw_{n}\}$, 
namely $\bfz = \sum_{i=1}^{n_0}\alpha_{i}\bfw_{i}$. Then:
\begin{align*}
\|L^{\text{sym}} \bfz\|^{2}_{2} = \|\sum_{i=1}^{n_0} \alpha_i\lambda_i\bfw_i\|^{2}_{2} = \|\sum_{i=2}^{n_0} 
\alpha_i\lambda_i\bfw_i\|^{2}_{2}  \geq \lambda_{2}^{2}\left(\sum_{i=2}^{n_0}\alpha_i^{2}\right) 
\end{align*}
and $\sum_{i=2}^{n_0}\alpha_i^{2} = \|\bfz\|_{2}^{2} - \alpha_{1}^{2}$. We now bound $\|\bfz\|_{2}$ and $\alpha_1$.
\begin{equation*}
\|\bfz\|_{2}^{2}  = \|D^{1/2}\bfv\|_{2}^{2} \geq \left(\sqrt{d_{\min}}\right)^{2}\|\bfv\|_{2}^{2} = d_{\min}
\end{equation*}
while:
\begin{equation*}
\alpha_1 = \langle \bfz,\bfw_1\rangle  = \langle D^{1/2}\bfv,\frac{1}{\sqrt{\text{vol}(G)}}D^{1/2}\bfone \rangle = \frac{1}{\sqrt{
\text{vol}(G)}}\langle \bfv,D\bfone\rangle \leq \frac{d_{\max}}{\sqrt{\text{vol}(G)}} \langle \bfv,\bfone\rangle.  
\end{equation*}
We now use the assumptions on $\bfv$. Specifically $\langle \bfv,\bfone\rangle \leq \|\bfv\|_{1} \leq \sqrt{s}\|\bfv\|_{2} = \sqrt{s}$ and so 
\begin{equation*}
\alpha_1 \leq d_{\max} \frac{\sqrt{s}}{\sqrt{\text{vol}(G)}} \leq  d_{\max} \frac{\sqrt{s}}{\sqrt{d_{\min}n_0}} = \frac{d_{\max}}{\sqrt{d_{\min}}}\frac{\sqrt{s}}{\sqrt{n_0}}.
\end{equation*}
Returning to equation \eqref{eq:LyBound}:
\begin{equation*}
\|L\bfv\|^{2}_{2} \geq \frac{1}{d_{\max}}\|L^{\text{sym}}\bfz\|_{2}^{2} \geq \frac{1}{d_{\max}}\lambda_{2}^{2}\left(d_{\min} - \frac
{d_{\max}^{2}}{d_{\min}}\frac{s}{n_0}\right) = 
\lambda_{2}^{2}\left( \frac{d_{\min}}{d_{\max}} - \frac{d_{\max}}{d_{\min}}\frac{s}{n_0}\right).   
\end{equation*}
These yield the desired estimate. 
\end{proof}

\begin{theorem}
Let $G\sim\mathcal{G}_n$ with $\mathcal{G}_n$ satisfying (A2) and (A4). Then for any $\gamma \in (0,1)$, we have that $\delta_{\gamma n_a}(L^{\text{in}}) \leq \frac{n_a}{n_1}\gamma + o(1)$.
\label{theorem:Bound_RIP_L_in}
\end{theorem}
\begin{proof}
Firstly, observe that $L^{\text{in}}$ is block diagonal with blocks $L_{G_{C_b}}$. For any block diagonal matrix we have that $\delta_{s}(L^{\text{in}}) = \max_{b}\delta_{s}(L_{G_{C_b}})$. By Lemma \ref{lemma:GenGraphDeltaBound} we have that:
\begin{equation}
\delta_{s}(L_{G_{C_b}}) \leq \max_{b}\{1 - \lambda_{2}(L_{G_{C_b}})^{2}\left(\frac{d^{\text{in}}_{\min}}{d^{\text{in}}_{\max}} - \frac{d^{\text{in}}_{\max}}{d^{\text{in}}_{\min}}\frac{s}{n_b}\right),
\lambda_{\max}(L_{G_{C_b}})^{2}-1\}.
\label{eq:L_G_C_a_RIP_bound}
\end{equation}
From assumption (A4) we get that:
\begin{equation*}
\frac{d^{\text{in}}_{\min}}{d^{\text{in}}_{\max}} =\frac{1-\epsilon_3}{1+\epsilon_3} = 1-o(1) \quad \text{ and } \quad \frac{d^{\text{in}}_{\max}}{d^{\text{in}}_{\min}} =\frac{1+\epsilon_3}{1-\epsilon_3} = 1+o(1). 
\end{equation*}
From assumption (A2) we get that:
\begin{equation*}
\lambda_{2}(L_{G_{C_b}})^{2} \geq  (1 - \epsilon_1)^{2} = 1 - 2\epsilon_1 + \epsilon_1^{2} = 1 - o(1)
\end{equation*}
and similarly $\lambda_{\max}(L_{G_{C_b}})^{2}-1 = o(1)$. Plugging this in to \eqref{eq:L_G_C_a_RIP_bound} with $s = \gamma n_a$ gives:
\begin{equation*}
\delta_{\gamma n_a}(L_{G_{C_b}})  \leq \max\left\{ \frac{\gamma n_a}{n_b} + o(1), o(1)\right\} \leq \gamma\frac{n_a}{n_1} + o(1) \ \implies \ 
\delta_{\gamma n_a}(L^{\text{in}}) \leq  \gamma\frac{n_a}{n_1} + o(1).
\end{equation*}
\end{proof}

\begin{remark}
\label{remark:Explaining_why_smallest_cluster}
We note that the RIP is only meaningful for $\delta_{\gamma n_a} < 1$. Hence the above theorem is only meaningful for $\gamma < \frac{n_1}{n_a} - o(1)$. To avoid this complicating technicality, we henceforth assume that $a = 1$, {\em i.e.} that the target cluster is $C_1$.
\end{remark}

\subsection{Bounding the size of the Perturbation}
 \begin{theorem}
Suppose that $G \sim \mathcal{G}_n$ with $\mathcal{G}_n$ satisfying (A3). If $L$ denotes the Laplacian of $G$ and $M:= L - L^{\text{in}}$ then $\|M\|_{2} \leq o(1)$.
\label{thm:MBound}
\end{theorem}
\begin{proof}
Letting $\delta_{ij}$ denote the Kronecker delta symbol, observe that 
\begin{equation*}
L_{ij} := \delta_{ij} - \frac{1}{d_i}A_{ij} = \delta_{ij} - \frac{1}{d^{\text{in}}_{i} + d^{\text{out}}_{i}}\left(A^{\text{in}}_{ij} + A^{\text{out}}_{ij}\right).
\end{equation*}
Earlier we defined $r_i = d_i^{\text{out}}/d_i^{\text{in}}$. We now use the following easily verifiable identity:
\begin{equation*}
\frac{1}{d^{\text{in}}_{i} + d^{\text{out}}_{i}} = \frac{1}{d^{\text{in}}_{i}} - \frac{1}{d^{\text{in}}_{i}}
\left(\frac{r_i}{r_i+1}\right).
\end{equation*}
Thus:
\begin{align*}
L_{ij} &= \delta_{ij} - \left(\frac{1}{d^{\text{in}}_{i}} - \frac{1}{d^{\text{in}}_{i}}
\left(\frac{r_i}{r_i+1}\right)\right)\left( A^{\text{in}}_{ij} + A^{\text{out}}_{ij}\right) \\
&= \left(\delta_{ij}  - \frac{1}{d^{\text{in}}_{i}}A^{\text{in}}_{ij}\right)  - \frac{1}{d^{\text{in}}_{i}}A^{\text
{out}}_{ij} + \frac{1}{d^{\text{in}}_{i}}\left(\frac{r_i}{r_i+1}\right)\left( A^{\text{in}}_{ij} + A^{\text{out}}_{ij}\right)  \\
	& = L^{\text{in}}_{ij} - \frac{1}{d^{\text{in}}_{i}}\left( 1 - \frac{r_i}{r_i+1}\right) A^{\text{out}}_{ij} + \frac{1}{d^{\text{in}}_{i}}\left(\frac{r_i}{r_i+1}\right)A^{\text{in}}_{ij} \\
	& = L^{\text{in}}_{ij} - \frac{1}{d^{\text{in}}_{i}}\left(\frac{1}{r_i+1}\right)A^{\text{out}}_{ij} + \frac{1}{d^{\text{in}}_{i}}\left(\frac{r_i}{r_i+1}\right)A^{\text{in}}_{ij}. 
\end{align*}
That is, $M_{ij} = - \frac{1}{d^{\text{in}}_{i}}\left(\frac{1}{r_i+1}\right)A^{\text{out}}_{ij} + 
\frac{1}{d^{\text{in}}_{i}}\left(\frac{r_i}{r_i+1}\right)A^{\text{in}}_{ij}$. 
To bound the spectral norm we use Gershgorin's disks, noting that $M_{ii} = 0$ for all $i$:
\begin{align*}
 \|M\|_{2} & = \max_{i}\{|\mu_{i}|: \ \mu_{i} \text{ eigenvalue of } M\} \leq \max_{i} \sum_{j} |M_{ij}|  \\
		 & = \max_{i}  \frac{1}{d^{\text{in}}_{i}}\left(\frac{1}{r_i+1}\right)\sum_{j} A^{\text{out}}_{ij} + \frac{1}{d^{\text{in}}_{i}}\left(\frac{r_i}{r_i+1}\right)\sum_{j}A^{\text{in}}_{ij}  \\
		 & = \max_{i} \left\{\frac{1}{d^{\text{in}}_{i}}\left(\frac{1}{r_i+1}\right)(d^{\text{out}}_{i}) + \frac{1}{d^{\text{in}}_{i}}\left(\frac{r_i}{r_i+1}\right)(d^{\text{in}}_{i})\right\}\\
		 & = \max_{i} \left\{ \left(\frac{r_i}{r_i+1}\right) + \left(\frac{r_i}{r_i+1}\right)\right\} \leq 2 \max_{i}r_{i} \leq 2\epsilon_{2} = o(1)
\end{align*}
by (A3).
\end{proof}

\begin{theorem}
Suppose that $G\sim\mathcal{G}_n$ with $\mathcal{G}_n$ satisfying (A1)--(A4). If $L$ denotes the Laplacian of $G$ and $|C_1\bigtriangleup \Omega| = \epsilon n_1$ with $\epsilon \leq 0.13$ then $\epsilon_{\bfy} = o(1)$ and $\epsilon^{ \gamma n_1}_{L} = o(1)$ for any $\gamma \in (0,1)$. 
\label{theorem:Bound_epsY_and_epsPhi}
\end{theorem}

\begin{proof}
Recall that $\displaystyle \epsilon_{\bfy} = \frac{\|\bfe\|_{2}}{\|\bfy^{\text{in}}\|_{2}} \text{ and } \epsilon^{\gamma n_1}_{L} = \frac{\|M\|_{2}^{(\gamma n_1)}}{\|L^{\text{in}}\|_{2}^{(\gamma n_1)}}$. Using the bound on the restricted isometry constant of $L^{\text{in}}$ from Theorem \ref{theorem:Bound_RIP_L_in}  we have:
\begin{align*}
\|\bfy^{\text{in}}\|^{2}_{2} &= \|L^{\text{in}}\left(\bfone_{W} - \bfone_{U}\right)\|_{2}^{2} \geq \left(1-\delta_{\epsilon n_1}(L^{\text{in}})\right)\|\bfone_{W} - \bfone_{U}\|_{2}^{2} \\
	& \geq \left(\epsilon -o(1)\right)|C_1\bigtriangleup\Omega| = (\epsilon^2 - o(1))n_1
\end{align*}

Thus $\|\bfy^{\text{in}} + \bfz^{\text{in}}\|_{2} \geq \sqrt{\epsilon^2 - o(1)}\sqrt{n_1}$. On the other hand:
$$
\|\bfe\|_{2} = \|\bfy - \bfy^{\text{in}}\|_{2} = \|L\bfone_{\Omega} - L^{\text{in}}\bfone_{\Omega}\|_{2} = \|M\bfone_{\Omega}\|_{2} \leq \|M\|_{2}\|\bfone_{\Omega}\|_{2} \leq o(1)\sqrt{(1+\epsilon)n_1} 
$$
Thus:
$$
\epsilon_{\bfy} = \frac{\|\bfe\|_{2}}{\|\bfy^{\text{in}}\|_{2}} \leq \frac{o(1)\sqrt{(1+\epsilon)}\sqrt{n_1}}{\sqrt{(\epsilon^{2} - o(1))}\sqrt{n_1}} = o(1)
$$
as $\epsilon$ is a constant, {\em i.e.} independent of $n_1$. The bound on $\epsilon^{\gamma n_1}_{L}$ is easier. By Lemma \ref{lemma:s_norm_bound_with_eigenvalues} and Property 3: 
$$
\|L^{\text{in}}\|_{2}^{(\gamma n_1)} \geq \sigma_{\gamma n_1 - 1}(L^{\text{in}}) = \lambda_{\gamma n_1 -1}(L^{\text{in}}) \geq \lambda_{k+1}(L^{\text{in}})
$$
as long as $\gamma n_1 \geq k+3$, which is certainly the case for large enough $n_1$. Because $\lambda_{1}(L_{G_{C_1}}) = \ldots = \lambda_{1}(L_{G_{C_k}}) = 0$ and the spectrum of $L^{\text{in}}$ is the union of the spectra of the $L_{G_{C_a}}$, it follows that:
$$
\lambda_{k+1}\left(L^{\text{in}}\right) = \min_{a=1}^{k}\lambda_{2}(L_{G_{C_a}}) \geq 1-\epsilon_1 = 1-o(1) 
$$
by (A1).
By Theorem \ref{thm:MBound} and Lemma \ref{lemma:s_norm_bound_with_eigenvalues} $\|M\|_{2}^{(\gamma n_1)} \leq \|M\|_2 =o(1)$. It follows that:
$$
\epsilon^{\gamma n_1}_{L} = \frac{\|M\|_{2}^{(\gamma n_1)}}{\|L^{\text{in}}\|_{2}^{(\gamma n_1)}} = \frac{o(1)}{1-o(1)} = o(1).
$$
\end{proof}

\subsection{Restricted Isometry Property for $L$}
Finally, we extend from $\delta_{s}(L^{\text{in}})$ to $\delta_{s}(L)$ using the following result of Herman and Strohmer (cf. \cite{HS10}):

\begin{theorem}
\label{thm:PerturbedRIC}
Suppose that $ \Phi = \hat{\Phi} + M$. Let $\hat{\delta}_{s}$ and $\delta_{s}$ denote the $s$ restricted isometry constants of $\hat{\Phi}$ and $\Phi$ respectively. Then:
\begin{equation*}
\delta_{s} \leq (1+\hat{\delta}_{s})\left(1+\epsilon^{s}_{\Phi}\right)^{2} - 1. 
\end{equation*}
\end{theorem}

\begin{corollary}
Let $L$ denote the Laplacian of $G\sim \mathcal{G}_n$ satisfying (A1)--(A4). Then we have 
$\delta_{\gamma n_1}(L) \leq \gamma + o(1)$ for any $\gamma \in (0,1)$. 
\end{corollary}
\begin{proof}
By Theorem \ref{thm:PerturbedRIC} we have that:
$$
\delta_{\gamma n_1}(L) \leq (1+\delta_{\gamma n_1}(L^{\text{in}}))(1 + \epsilon^{\gamma n_1}_{L})^{2} - 1.
$$
Substituting the values of  $\delta_{\gamma n_1}(L^{\text{in}})$ and $\epsilon^{\gamma n_1}_{L}$ from Theorems \ref{theorem:Bound_RIP_L_in} and \ref{theorem:Bound_epsY_and_epsPhi} yields the claim.
\end{proof}

\section{Proof of Theorem \ref{thm:FindingOmegaSSBM}\label{A:Part1Proof}}
Before proving this theorem we prove the a series of  lemmas. We first note that Assumptions (A3) and (A4) easily allow us to bound $\text{vol}(S)$, which will be required in the proof of Theorem \ref{thm:FindingOmegaSSBM}:

\begin{lemma}
\label{lemma:BoundVolume}
Suppose that $\mathcal{G}_n$ satisfies (A3) and (A4). For any $S\subset V$ define 
$\text{vol}^{\text{in}}(S) = \sum_{i} d^{\text{in}}_i$. Then for any $G\in \mathcal{G}_n$ we 
have that:
\begin{equation*}
(1)\  \displaystyle (1-\epsilon_3)|S|d^{\text{in}}_{\text{av}} \leq \text{vol}^{\text{in}}(S) 
\leq (1+\epsilon_3)|S|d^{\text{in}}_{\text{av}}; \hbox{ and } 
(2) \ \displaystyle \text{vol}^{\text{in}}(S) \leq \text{vol}(S) \leq (1+\epsilon_2)\text{vol}^{\text{in}}(S). 
\end{equation*}
\end{lemma}

\begin{proof}
For part (1), observe that:
\begin{equation*}
\text{vol}^{\text{in}}(S) = \sum_{i\in S} d_{i}^{\text{in}} \geq |S|d^{\text{in}}_{\min} 
\geq |S|(1-\epsilon_3)d^{\text{in}}_{\text{av}}, 
\end{equation*}
where the final inequality is from (A4). The bound $\text{vol}^{\text{in}}(S) \leq 
(1+\epsilon_3)|S|d^{\text{in}}_{\text{av}}$ follows similarly. For part (2) we note that by 
assumption (A3) $d_{i} = d_{i}^{\text{in}} + d_{i}^{\text{out}} \leq d_{i}^{\text{in}} + 
\epsilon_2d_{i}^{\text{in}} = (1+\epsilon_2)d_{i}^{\text{in}}$. Hence:
$$
\text{vol}(S)  = \sum_{i\in S} d_{i} \leq \sum_{i\in S} (1+\epsilon_2)d_{i}^{\text{in}} = 
(1+\epsilon_2)\text{vol}^{\text{in}}(S)
$$
while the lower bound follows simply from the fact that $d_{i} \geq d_{i}^{\text{in}}$.
\end{proof}

\begin{lemma}
\label{lemma:BoundInnerProd}
Let $G\in \mathcal{G}_n$ satisfies Assumptions (A1)--(A4). If $N_{G_{C_1}}:= D_{G_{C_1}}^{-1/2} A_{G_{C_1}} D_{G_{C_{1}}}^{-1/2}$ and $U,\Gamma \subset C_1$ then:
$$
\left| \langle D^{1/2}_{G_{C_1}}\bfone_{U},N^{t}_{G_{C_1}}D^{1/2}_{G_{C_1}}\bfone_{\Gamma}\rangle - \frac{\text{vol}^{\text{in}}(U)\text{vol}^{\text{in}}(\Gamma)}{\text{vol}^{\text{in}}(G_{C_1})} \right| \leq \epsilon_1^{t}\sqrt{\text{vol}^{\text{in}}(U)\text{vol}^{\text{in}}(\Gamma)}
$$
\end{lemma}
\begin{proof}
From the proof of Lemma 2 in \cite{CG08} (note that they use $M_{G_{C_1}}$ instead of $N_{G_{C_1}}$) we get that:
 $$
\left| \langle D^{1/2}_{G_{C_1}}\bfone_{U},N^{t}_{G_{C_1}}D^{1/2}_{G_{C_1}}\bfone_{\Gamma}\rangle - \frac{\text{vol}^{\text{in}}(U)\text{vol}^{\text{in}}(\Gamma)}{\text{vol}^{\text{in}}(G_{C_1})} \right| \leq \lambda_{n_1-1}(N_{G_{C_1}})^{t}\sqrt{\text{vol}^{\text{in}}(U)\text{vol}^{\text{in}}(\Gamma)}
$$
By Lemma \ref{lemma:Evals_of_L_Lsym_P} and (A2) we get that $\lambda_{n_1-1}(N_{G_{C_1}}) = 1 - \lambda_2(L_{G_{C_1}}) \leq \epsilon_1$. 
\end{proof}

\begin{proof}[Proof of Theorem \ref{thm:FindingOmegaSSBM}]
As in \S \ref{sec:ClusterPursuit}, let $U = C_1\setminus\Omega$ and $W = \Omega\setminus C_1$. Let $|U| = un_1$, in which case $|W| = (\epsilon + u)n_1$. We shall prove that $u = o(1)$. By definition, $\Omega$ is the set of the $(1+\epsilon)n_1$ largest entries in $\bfv^{(t)} := P^{t}D\bfone_{\Gamma}$. Because $U$ is not in $\Omega$, but $W$ is, we must have $v^{(t)}_{i} \leq  v^{(t)}_{j}$ for every $i\in U$ and $j \in W$. We sum first over $j\in W$ and then sum over $i\in U$ to obtain:
\begin{equation*}
v^{(t)}_{i} \leq  v^{(t)}_{j} \implies (\epsilon + u)n_1 v^{(t)}_i \leq \sum_{j\in W} v^{(t)}_{j} \Longrightarrow (\epsilon + u)n_1\sum_{i\in U} v^{(t)}_i \leq un_1\sum_{j\in W} v^{(t)}_{j}.
\end{equation*}
It follows that:
\begin{equation}
\sum_{i\in U} v^{(t)}_i \leq \frac{u}{\epsilon + u} \sum_{j\in W} v^{(t)}_{j} \leq \sum_{j\in W} v^{(t)}_{j}.  
\label{eq:BreakInequality}
\end{equation}
Looking ahead, we shall show that if inequality \eqref{eq:BreakInequality} holds then $u = o(1)$. \\

We first show that the term on the left-hand side of inequality \ref{eq:BreakInequality}, {\em i.e.} the sum over the vertices in $C_{1}$ that were missed by $\Omega$, is necessarily quite large. We do this by relating $P$ to $P^{\text{in}}$, the random walk transition matrix for the graph $G^{\text{in}}$. Note that $G^{\text{in}}$ is a disjoint union of the graphs $G_{C_{a}}$. For every $i\in [n]$, define $q_{i} := d^{\text{in}}_{i}/d_{i}$. Observe that $1/d_{i} = q_i/d^{\text{in}}_{i}$ and thus $D^{-1} = D_{\text{in}}^{-1}Q$ where $Q$ is the diagonal matrix with $(i,i)$-th entry $q_i$. Now:
$$
P = AD^{-1} = \left(A^{\text{in}} + A^{\text{out}}\right)D^{-1} = A^{\text{in}}\left(D_{\text{in}}^{-1}Q\right) + A^{\text{out}}D^{-1} = P^{\text{in}}Q + A^{\text{out}}D^{-1}.
$$ 
Observe that $P$, $P^{\text{in}}Q$ and $A^{\text{out}}D^{-1}$ all have non-negative entries. It follows that for any non-negative vector $\bfx$:
$P\bfx$ and $P^{\text{in}}Q\bfx$ are also non-negative and 
$P\bfx \geq P^{\text{in}}Q\bfx$, where the inequality should be interpreted componentwise.
One can the extend the inequality by iterated multiplication:
$$
P^{t}\bfx \geq \left(P^{\text{in}}Q\right)^{t}\bfx \geq q^{t}_{\min}\left(P^{\text{in}}\right)^{t}\bfx
$$
and again the inequality should be interpreted componentwise. Now:
\begin{align}
\sum_{i\in U} v^{(t)}_i & = \langle \bfone_{U}, \bfv^{(t)} \rangle = \langle \bfone_{U}, P^{t}D\bfone_{\Gamma} \rangle \geq \langle \bfone_{U}, q^{t}_{\min}\left(P^{\text{in}}\right)^{t}D\bfone_{\Gamma}\rangle \nonumber  \\
& = q_{\text{min}}^{t}\langle \bfone_{U},\left(P^{\text{in}}\right)^{t}D_{\text{in}}\bfone_{\Gamma}\rangle \nonumber = q_{\text{min}}^{t}\langle \bfone_{U},\left(P_{G_{C_1}}\right)^{t}D_{G_{C_1}}\bfone_{\Gamma}\rangle,
\label{eq:LowerBoundOnUSum}
\end{align}
where the final line follows as $U,\Gamma \subset C_{1}$. 

Our goal now is to bound the quantity $\langle \bfone_{U},\left(P_{G_{C_1}}\right)^{t}D_{G_{C_1}}\bfone_{\Gamma}\rangle$. One can rearrange the iterated matrix product slightly:
\begin{align*}
\left(P_{G_{C_1}}\right)^{t} & = \left(A_{G_{C_1}}D^{-1}_{G_{C_1}}\right)^{t} = A_{G_{C_1}}D^{-1}_{G_{C_1}}A_{G_{C_1}}D^{-1}_{G_{C_1}} \ldots A_{G_{C_1}}D^{-1}_{G_{C_1}} \\
	& = D^{1/2}_{G_{C_1}}\left(D^{-1/2}_{G_{C_1}}A_{G_{C_1}}D^{-1/2}_{G_{C_1}}\right)
	\left(D^{-1/2}_{G_{C_1}}A_{G_{C_1}}D^{-1/2}_{G_{C_1}}\right) \ldots \left(D^{-1/2}_{G_{C_1}}A_{G_{C_1}}D^{-1/2}_{G_{C_1}}\right)D^{-1/2}_{G_{C_1}} \\
	& = D^{1/2}_{G_{C_1}}N_{G_{C_1}}^{t}D^{-1/2}_{G_{C_1}},
\end{align*}

Hence, we have 
\begin{align*}
\langle \bfone_{U},\left(P_{G_{C_1}}\right)^{t}D_{G_{C_1}}\bfone_{\Gamma}\rangle &= \langle \bfone_{U},\left(D^{1/2}_{G_{C_1}}N_{G_{C_1}}^{t}D^{-1/2}_{G_{C_1}}\right)D_{G_{C_1}}\bfone_{\Gamma}\rangle \\
& = \langle D^{1/2}_{G_{C_1}}\bfone_{U},N_{G_{C_1}}^{t}D^{1/2}_{G_{C_1}}\bfone_{\Gamma}\rangle 
 \geq \frac{\text{vol}^{\text{in}}(U)\text{vol}^{\text{in}}(\Gamma)}{\text{vol}^{\text{in}}(G_{C_1})} - \epsilon_1^{t}\sqrt{\text{vol}^{\text{in}}(U)\text{vol}^{\text{in}}(\Gamma)},
\end{align*}
where the final inequality follows from Lemma \ref{lemma:BoundInnerProd}. Returning to \eqref{eq:LowerBoundOnUSum}:
\begin{equation}
\sum_{i\in U} v^{(t)}_i \geq q_{\min}^{t}\left(\frac{\text{vol}^{\text{in}}(U)\text{vol}^{\text{in}}(\Gamma)}{\text{vol}^{\text{in}}(G_{C_1})} - \epsilon_1^{t}\sqrt{\text{vol}^{\text{in}}(U)\text{vol}^{\text{in}}(\Gamma)}\right).
\label{eq:BoundInnerProdWithU}
\end{equation}
We now consider the right hand side of \eqref{eq:BreakInequality}, {\em i.e.} the sum over $W$. 
Because $W\subset V\setminus C_{1}$ we have that:
$$
\sum_{j\in W} v^{(t)}_{j} \leq \sum_{j\in V\setminus C_1}|v^{(t)}_{j}| = \|\bfv^{(t)}_{V\setminus C_1}\|_{1}
$$
Thus it remains to bound $\|\bfv_{V\setminus C_1}^{(t)}\|_{1}$. Observe that:
$$
\bfv_{V\setminus C_1}^{(t)} = A^{\text{in}}D^{-1}\bfv_{V\setminus C_1}^{(t-1)} + \left(A^{\text{out}}D^{-1}\bfv^{(t-1)}\right)_{V\setminus C_1}.
$$
Clearly
$$
\left\|\left(A^{\text{out}}D^{-1}\bfv^{(t-1)}\right)_{V\setminus C_1}\right\|_{1} \leq \left\|A^{\text{out}}D^{-1}\bfv^{(t-1)}\right\|_{1}
$$
and so
$$
\|\bfv_{V\setminus C_1}^{(t)}\|_{1} \leq \|A^{\text{in}}D^{-1}\bfv_{V\setminus C_1}^{(t-1)}\|_{1} + \|A^{\text{out}}D^{-1}\bfv^{(t-1)}\|_{1} \leq \|A^{\text{in}}D^{-1}\|_{1}\|\bfv_{V\setminus C_1}^{(t-1)}\|_{1} + \|A^{\text{out}}D^{-1}\|_{1}\|\bfv^{(t-1)}\|_{1}
$$

Moreover: $\|A^{\text{in}}D^{-1}\|_{1} = \max_{j} \sum_{i}\frac{A^{\text{in}}_{ij}}{d_{j}} = \max_{j}\frac{d^{\text{in}}_{j}}{d_{j}} \leq 1$ and similarly $\|A^{\text{out}}D^{-1}\|_{1} = \max_{j}\frac{d^{\text{out}}_{j}}{d_j} \leq \max_{j}r_{j} \leq \epsilon_2$ by assumption (A2). Thus $
\|\bfv_{V\setminus C_1}^{(t)}\|_{1} \leq 1\|\bfv_{V\setminus C_1}^{(t-1)}\|_{1} + \epsilon_2\|\bfv^{(t-1)}\|_{1}
$. Solving this recursion relation we obtain:
$$
\|\bfv_{V\setminus C_1}^{(t)}\|_{1} \leq \epsilon_2 \sum_{s=0}^{t-1}\|\bfv^{(s)}\|_1 + \|\bfv_{V\setminus C_1}^{(0)}\|_{1}
$$
Because $\bfv^{(0)} = D\bfone_{\Gamma}$ and $\Gamma\subset C_1$, it follows that $\|\bfv^{(0)}_{V\setminus C_{1}}\|_{1} = 0$ and $\|\bfv^{(0)}\|_{1} = \text{vol}(\Gamma)$. Because $\|P\|_1 = 1$ it follows that $\|\bfv^{(s)}\|_1 = \|\bfv^{(0)}\|_1 = \text{vol}(\Gamma)$ for all $s$. Thus:

\begin{equation}
\sum_{j\in W} v^{(t)}_{j} \leq \|\bfv_{V\setminus C_1}^{(t)}\|_{1} \leq t\epsilon_2\text{vol}(\Gamma) \leq t\epsilon_2(1+\epsilon_2)\text{vol}^{\text{in}}(\Gamma),
\label{eq:BoundInnerProdWithW}
\end{equation}
where the final inequality follows from Lemma \ref{lemma:BoundVolume}. Now let us put this all together. Returning to \eqref{eq:BreakInequality} with \eqref{eq:BoundInnerProdWithU} and \eqref{eq:BoundInnerProdWithW} in hand:
\begin{align}
& q_{\min}^{t}\left(\frac{\text{vol}^{\text{in}}(U)\text{vol}^{\text{in}}(\Gamma)}{\text{vol}^{\text{in}}(G_{C_1})} - \epsilon_1^{t}\sqrt{\text{vol}^{\text{in}}(U)\text{vol}^{\text{in}}(\Gamma)}\right) \leq t\epsilon_2(1+\epsilon_2)\text{vol}^{\text{in}}(\Gamma)  \\
\Longrightarrow & q_{\min}^{t}\left(\frac{\text{vol}^{\text{in}}(U)}{\text{vol}^{\text{in}}(G_{C_1})} - \epsilon_1^{t}\sqrt{\frac{\text{vol}^{\text{in}}(U)}{\text{vol}^{\text{in}}(\Gamma)}}\right) \leq t\epsilon_2(1+\epsilon_2). \label{eq:BalanceBound}
\end{align}

From Lemma \ref{lemma:BoundVolume} and the assumptions on $|U|$ and $|\Gamma|$:
\begin{align*}
& \frac{\text{vol}^{\text{in}}(U)}{\text{vol}^{\text{in}}(G_{C_1})} \geq \frac{(1-\epsilon_3)d^{\text{in}}_{\text{av}}|U|}{(1+\epsilon_3)d^{\text{in}}_{\text{av}}|C_1|} = \frac{(1-\epsilon_3)un_1}{(1+\epsilon_3)n_1} = \frac{1-\epsilon_3}{1+\epsilon_3}u \\
& \frac{\text{vol}^{\text{in}}(U)}{\text{vol}^{\text{in}}(\Gamma)} \leq \frac{(1+\epsilon_{3})d^{\text{in}}_{\text{av}}|U|}{(1-\epsilon_3)d^{\text{in}}_{\text{av}}|\Gamma|} \leq \frac{(1+\epsilon_3)}{(1-\epsilon_3)}\frac{u}{g\epsilon_{1}^{2t-1}}
\end{align*} 
Finally because $q_i = 1 - r_i$ it follows that $q_{\min} \geq 1 - \epsilon_2$. Putting this all into equation \eqref{eq:BalanceBound}:
$$
(1-\epsilon_2)^{t}\left(\frac{1-\epsilon_3}{1+\epsilon_3}u - \epsilon_1^{1/2}\sqrt{\frac{(1+\epsilon_3)}{(1-\epsilon_3)}\frac{u}{g}}\right) \leq t\epsilon_2(1+\epsilon_2) 
$$
At this stage it is illuminating to use the assumption that $\epsilon_1,\epsilon_2,\epsilon_3 = o(1)$. Observe that:
$$
\frac{1-o(1)}{1+o(1)} = 1-o(1), \quad \frac{1+o(1)}{1-o(1)} = 1+o(1), \ \text{ and } (1-o(1))^{t} = 1-o(1)
$$
where the final equality follows as $t$ is constant with respect to $n$. Hence:
$$
(1-o(1))u - o(\sqrt{u}) \leq o(1) \Longrightarrow u \leq o(1) + o(u).
$$
This is only possible if $u = o(1)$. It follows that $|C_1\bigtriangleup\Omega| = |U| + |W| = (\epsilon + 2u)n_1 = (\epsilon + o(1))n_1$ as stated. 
\end{proof}

\section{Showing the SBM satisfies our assumptions\label{A:SBM}}
\label{section:Assumptions_for_SBM}

Here we verify that $\text{SBM}(\mathbf{n},P)$ satisfies the assumptions (A1)--(A4), under the hypotheses of Theorem \ref{theorem:Assumptions_for_SBM}. Recall that we are assuming that $P_{ab} = (\beta+o(1))\log(n)/n$ for $a\neq b$, and that $P_{aa} = \omega\log(n)/n_a$ for $a = 1,\ldots, k$. As we are also assuming that $n_1\to \infty$, and $n_1$ is the size of the smallest cluster, we get that $k = O(1)$, {\em i.e.} (A1) holds.

\begin{theorem}[see \cite{B82, B01}]
\label{thm:Concentration_q}
Let $G \sim \text{ER}(n,q)$ with $q = (\beta+o(1))\log(n)/n$. There exist a function 
$\eta(\beta)$ satisfying $0 < \eta(\beta) < 1$ and $ \lim_{\beta\to\infty} \eta(\beta) = 0$ such that 
\begin{equation*}
d_{\text{max}}(G) = (1+\eta(\beta))\beta\log{n} + o(1) \leq 2\beta\log(n) + o(1) \text{ a.s. }
\end{equation*}
\end{theorem}

\begin{theorem}[see \cite{FK16}, Theorem 3.4 (ii)]
If  $G \sim \text{ER}(n_a,p)$ with $p_a = \omega\log(n)/n_a$ where $\omega \to \infty$, then $d_{\min}(G) = (1- 
o(1))\omega\log(n)$ and $d_{\max}(G) = (1+o(1))\omega\log(n)$ a.s.
\label{thm:Concentration_p}
\end{theorem}

\begin{theorem}
\label{thm:EigConcentration}
Suppose that $G \sim \text{ER}(n_a,p)$ with $p = \omega \log(n)/n_a$ where $\omega\to\infty$. Then we have almost surely $|\lambda_{i}(L) - 1|  = O(\omega^{-1/2}) = o(1)$ for all $i  > 1$.
\end{theorem}
\begin{proof}
Theorem 4 in \cite{CR11} shows that
$$
|\lambda_{i}(L^{\text{sym}}) - 1| \leq \sqrt{\frac{6\log(2n_a)}{\omega \log(n)}}.
$$
By Lemma \ref{lemma:Evals_of_L_Lsym_P} $L^{\text{sym}}$ and $L$ have the same spectrum. The result follows as $\log(n) \geq \log(n_a)$
\end{proof}

As each $G_{C_a} \sim \text{ER}(n_a,p)$, it follows from Theorem \ref{thm:EigConcentration} that:
\begin{corollary}
	$\text{SBM}(\mathbf{n},P)$ with parameters as in Theorem \ref{theorem:Assumptions_for_SBM} satisfies assumption (A2) with $\epsilon_1 = O(\omega^{-1/2})$.
\end{corollary}

We now discuss the remaining two assumptions. Let $G^{\text{in}}$ and $G^{\text{out}}$ be as in \S \ref{sec:Preliminaries}. If 
$G \sim \text{SBM}(\mathbf{n},P)$ then $G^{\text{in}}$ consists of $k$ disjoint \ER graphs,  $G_{C_a} \sim \text{ER}(n_a,p)$. The graph $G^{\text{out}}$ is not an \ER graph, as there is zero probability of it containing an 
edge between two vertices in the same cluster (because we have removed them). However, we can profitably think of $G^{\text{out}}$ 
as a subgraph of some $\widetilde{G^{\text{out}}} \sim \text{ER}(n,q)$. In particular, any upper bounds on the degrees of vertices 
in $\widetilde{G^{\text{out}}}$ are automatically bounds on the degrees in $G^{\text{out}}$. Thus, we have the following 
corollaries of Theorems \ref{thm:Concentration_p} and \ref{thm:Concentration_q}:

\begin{corollary}
\label{cor1}
If $G \sim \text{SBM}(\mathbf{n},P)$ with parameters as in Theorem \ref{theorem:Assumptions_for_SBM} then $d^{\text{out}}_{\max}(G) \leq 2\beta\log{n} + o(1)$ a.s.
\end{corollary}
\begin{proof}
Consider $G^{\text{out}}$ as a subgraph of $\widetilde{G^{\text{out}}} \sim \text{ER}(n,q)$ and apply Theorem \ref{thm:Concentration_q}
\end{proof}

\begin{corollary}
\label{cor:Bounding_d_in}
If $G \sim \text{SBM}(\mathbf{n},P)$ with parameters as in Theorem \ref{theorem:Assumptions_for_SBM}, then $d^{\text{in}}_{\min}(G) \geq 
(1-o(1))\omega\log(n)$ and $d^{\text{in}}_{\max}(G) \leq (1+o(1))\omega\log(n)$ a.s.
\end{corollary}
\begin{proof}
If $i\in C_a$ then $d^{\text{in}}_{i} = d_{i}(G_{C_a})$, where $G_{C_a} \sim \text{ER}(n_a,p)$. Clearly:
\begin{equation*}
d^{\text{in}}_{\text{max}}(G) = \max_{i} d^{\text{in}}_{i} = \max_{a} d_{\text{max}}(G_{C_a})
\end{equation*}
By Theorem \ref{thm:Concentration_p}, $d_{\text{max}}(G_a) = (1+o(1))\omega\log(n)$ a.s. Note that 
the $d_{\text{max}}(G_{C_a})$ are independent random variables, and since we are taking a maximum over $k = \mathcal{O}(1)$ of them, it 
follows that $\max_a d_{\text{max}}(G_{C_a}) \leq (1+o(1))\omega\log(n)$ a.s. too. The proof for $d^{\text{in}}_{\min}(G)$ is similar.
\end{proof}

\begin{corollary}
\label{cor:Bounding_r}
$\text{SBM}(\mathbf{n},P)$ with parameters as in Theorem \ref{theorem:Assumptions_for_SBM} satisfies assumption (A3) with $\epsilon_2 = O(\omega^{-1})$.
\end{corollary}
\begin{proof}
First of all, it is clear that for any $i$, $d^{\text{out}}_{i}/d^{\text{in}}_{i} \leq d^{\text{out}}_{\max}/d^{\text{in}}_{\min}$. From Corollaries \ref{cor1} and \ref{cor:Bounding_d_in} we have:
\begin{align*}
\frac{d^{\text{out}}_{\max}}{d^{\text{in}}_{\min}} & \leq \frac{2\beta\log{n} + o(1)}{(1- o(1))\omega\log(n)} =  \frac{2\beta + o(1)}{(1-o(1))\omega} = O(\omega^{-1}).
\end{align*} 
\end{proof}

\begin{corollary}
$\text{SBM}(\mathbf{n},P)$ with parameters as in Theorem \ref{theorem:Assumptions_for_SBM} satisfies assumption (A4).
\end{corollary}
\begin{proof}
Observe that $d^{\text{in}}_{\text{av}} = \omega\log(n)$. The result then follows from Corollary \ref{cor:Bounding_d_in}.
\end{proof}

\section{Implementation of Algorithms}
\label{Supplement:Implementation}
All numerical experiments were done in {\tt MATLAB} on a mid 2012 Macbook pro with a 2.5 GHz Intel Core i5 processor and 16 GB of RAM. \\

\noindent
{\tt FlowImprove} We use an implementation available at \url{https://dgleich.wordpress.com/2011/09/19/fast-partition-improvement-with-flowimprove/} that uses the MATLAB-BGL package available at \url{https://github.com/dgleich/matlab-bgl}. We are extremely grateful to D. Gleich for some assistance in getting MATLAB-BGL to run on Mac OS X. {\tt FlowImprove} has two parameters: a vertex weighting vector $\mathbf{p}\in\mathbb{R}^{n}$ and a maximum number of iterations. We keep both at their default values, namely the all-ones vector and 5 respectively, for all experiments.\\

\noindent
{\tt SimpleLocal} We use the implementation available at \url{https://github.com/nveldt/SimpleLocal} which uses the Gurobi (\url{https://www.gurobi.com/}) optimization package in the maxflow subroutine. {\tt SimpleLocal} has one parameter, a locality parameter $\delta$. \\

\noindent
{\tt HKGrow} We use the implementation of this algorithm available at 
\url{https://www.cs.purdue.edu/homes/dgleich/codes/hkgrow/}. This implementation requires no input parameters.  \\

\noindent
{\tt PPR-Grow} We use the implementation that is available at \url{https://www.cs.purdue.edu/homes/dgleich/codes/hkgrow/}. Again, note that this is C++ code with a MATLAB wrapper, so we expect it to be faster than MATLAB-only code. This algorithm has a teleportation parameter, $\alpha$, and a tolerance parameter, $\epsilon$.  \\

\noindent
{\tt LBSA} We use the MATLAB implementation provided by the authors of \cite{SHBH19}, available at \url{https://github.com/PanShi2016/LBSA}.  The {\tt LBSA} algorithm actually includes six distinct methods; we use the heat kernel sampling with Lanczos method, denoted in \cite{SHBH19} as {\tt hkLISA}, as experimental evidence presented in the aforementioned paper suggests that this variant performs best. We also tried other methods (specifically heat kernel sampling with power method, and random walk sampling with power and Lanczos methods), but did not observe any significant difference in performance on our data sets. This algorithm requires one parameter, $k_2$, which governs the number of Lanczos iterations to take.\\

{\tt ClusterPursuit}, {\tt CP+RWT} and {\tt ICP+RWT} are all written in MATLAB and available as the ``ClusterPursuit'' package from the second author's website. \\

Note that the implementations of {\tt FlowImprove}, {\tt HKGrow} and {\tt PPR-Grow} used are all written in C++  and run in MATLAB using the {\tt mex} API. As such, we expect these implementations to run  several times faster than MATLAB-only implementations of these algorithms.

\section{Parameters for Numerical Experiments}
\label{sec:NumericalParameters}
\subsection{Synthetic Data}
For both local clustering experiments, {\em i.e.} using $\text{SBM}(\mathbf{n}^{(1)},P^{(1)})$ and $\text{SBM}(\mathbf{n}^{(2)},P^{(2)})$, we use the same parameters. For {\tt PPR-Grow}, following the discussion in \S C.2 of \cite{W17}, we try several values of $\alpha$ in the range $[\lambda/2,2\lambda]$ where $\lambda$ is the smallest non-zero eigenvalue of $L$. We observe best performance for $\alpha = \lambda$ so we use this value. For {\tt LBSA} we use $k_2 = 4$, as suggested in \cite{SHBH19}. For {\tt CP+RWT} we use parameters that align with Theorem 5.1, namely $\epsilon = 0.13/2$, $s = 0.13 n_1$, $R = 0.5$, $t = 3$ and $\hat{n}_1 = n_1$. For all algorithms we use the same seed set, $\Gamma$, drawn uniformly at random from $C_1$ and of size $|\Gamma| = 0.01n_1$. We make no attempt to tune parameters here for {\tt CP+RWT}, and note that one can get even better performance by choosing larger values of $\epsilon$ and $s$. \\

For the cut improvement experiments, we generate an initial cut using {\tt RWThresh} with parameters $\epsilon = 0.13$ $\hat{n}_1 = n_1$, $t = 3$ and $\Gamma$ chosen uniformly at random from $C_1$ with $|\Gamma| = 0.01n_1$. For {\tt SimpleLocal} we set the parameter $\delta$ to $0.5$. For {\tt ClusterPursuit} we experimented with various values of $s$, and reported results for $s = 0.26n_1$ for $\text{SBM}(\mathbf{n}^{(1)},P^{(1)})$ and $s = 0.16n_1$ for $\text{SBM}(\mathbf{n}^{(2)},P^{(2)})$. We fix $R = 0.5$ and $\hat{n}_1 = n_1$. 

\subsection{Social Networks}
For {\tt LBSA}, we again use $k_2 = 4$. For {\tt PPR-Grow}, we tune $\alpha$ to  $\alpha = 4\lambda$, where again $\lambda$ is the smallest non-zero eigenvalue of $L$. For {\tt CP+RWT} we take $\epsilon = 0.25$, $R = 0.5$, $t = 3$ and $s = 0.5n_1$.  

\subsection{MNIST and OptDigits}
For both MNIST and Optdigits, we take $\epsilon = 0.13$, $R = 0.5$, $t=3$, $s = 0.26n_1$ and $\hat{n}_a = n_a$.

\section{Preprocessing Image Data}
\label{sec:PreprocessingImages}
For MNIST and OptDigits, we construct a weighted $K$-NN graph as follows. Note that $K$, the number of neighbors, has no relation to $k$, the number of clusters.

\begin{itemize}
\item Let $\mathcal{X} = \{\bfx_1,\ldots, \bfx_{n}\}$ denote the vectorized version of the data set. That is, if the data set consists of $8\times 8$ images then $\mathcal{X}\subset \mathbb{R}^{64}$.
\item Fix parameters $r=10$ and $K=15$.
\item For all $i\in [n]$, define $\sigma_i:= \|\bfx_i - \bfx_{[r,i]}\|$, where $\bfx_{[r,i]}$ denotes the $r$-th closest point in $\mathcal{X}$ to $\bfx_i$. (If there is a tie, break it arbitrarily). Let $\text{NN}(\bfx_i,K)\subset \mathcal{X}$ denote the set of the $K$ closest points in $\mathcal{X}$ to $\bfx_i$. Again, one may break ties arbitrarily if they occur.
\item Define $\tilde{A}$ as: $\tilde{A}_{ij} = \left\{\begin{array}{cc}\exp\left(-\|\bfx_i - \bfx_j\|^{2}/\sigma_i\sigma_j\right) & \text{ if } \bfx_j \in NN(\bfx_i,K)\\ 0 & \text{otherwise} \end{array}\right.$
\item Observe that $\tilde{A}$ is not necessarily symmetric, as it may occur that $\bfx_j \in \text{NN}(\bfx_i,K)$ while $\bfx_i \notin \text{NN}(\bfx_j,K)$. So, we take $A = \tilde{A}^{\top}\tilde{A}$ to be the adjacency matrix that we use in our experiments
\end{itemize}

\end{appendices}

\end{document}